\documentclass[final,1p,times]{elsarticle}

\usepackage{epsfig}

\usepackage{amsfonts,amssymb,amscd,amsmath,amsthm,enumerate,verbatim,calc}
\usepackage{graphicx}

\usepackage{amssymb}

\usepackage{amsthm}
\usepackage{enumitem}
\usepackage{algorithmicx}
\usepackage{algorithm}
\usepackage{algc}
\usepackage{algcompatible}
\usepackage{algpseudocode}
\usepackage{makecell}
\usepackage{titlesec}
\usepackage[utf8]{inputenc}
\renewcommand{\algorithmicrequire}{\textbf{Input:}}
\renewcommand{\algorithmicensure}{\textbf{Output:}}

\algnewcommand{\LINECOMMENT}[1]{\STATE\(\triangleright\) #1}

\makeatletter
\newcommand{\mathleft}{\@fleqntrue\@mathmargin20pt}
\newcommand{\mathcenter}{\@fleqnfalse}

\makeatother

\theoremstyle{definition}
\newtheorem{Theorem}{Theorem}[section]

\newtheorem{Corollary}{Corollary}[section]
\newtheorem{Proposition}{Proposition}[section]
\newtheorem{Remark}{Remark}[section]

\newtheorem{Example}{Example}[section]

\newtheorem{defn}{Definition}[section]

\newtheorem*{prf}{Proof}

\begin{document}

\begin{frontmatter}
\title{Metric-Based Granular Computing in Networks}
%
\author[1]{Hibba Arshad}
\affiliation[1]{organization={Centre for Advanced Studies in Pure and Applied Mathematics},
            addressline={Bahauddin Zakariya University},
            city={Multan},
            postcode={6000},
            country={Pakistan.}}
\ead{hibbaarshad87@gmail.com}
\cortext[cor1]{Corresponding author}
\author[1]{Imran Javaid\corref{cor1}}
\ead{imran.javaid@bzu.edu.pk}
\begin{abstract}
Networks can be highly complex systems with numerous interconnected components and interactions. Granular computing offers a framework to manage this complexity by decomposing networks into smaller, more manageable components or granules. In this article, we introduce metric-based granular computing technique to study networks.
This technique can be applied to the analysis of networks where granules can represent subsets of nodes or edges and their interactions can be studied at different levels of granularity. We model the network as an information system and investigate its granular structures using metric representation. We establish that the concepts of reducts in rough set theory and resolving sets in networks are equivalent. Through this equivalence, we present a novel approach for computing all the minimal resolving sets of these networks.

\end{abstract}
\end{frontmatter}
\section{Introduction}
Networks provide a framework for modeling and analyzing relationships and interactions in various fields. A network consists of interconnected elements, typically referred to as nodes or vertices, linked by edges or connections. Networks are commonly used to represent systems where entities interact, such as social networks, where individuals are connected through friendships, communication, or collaborations.
Granular computing is an information processing framework that simplifies complex networks, enabling easier analysis and understanding of their structural dynamics and relationships. Granular computing was initially proposed by Zadeh in 1979 within the framework of fuzzy set theory \cite{zadeh}. Granular computing is an emerging knowledge retrieval model, that handles information processing at different levels of granularity, focusing on the organization of complex entities, known as information granules.
These granules represent groups of individuals in a network who share similar features or behaviors. Granular computing has many applications in network analysis, providing solutions to problems in various domains, including social networks \cite{lisenti, raj}. One well-known application of granular computing in social network analysis is the classification and clustering of nodes according to their connectivity patterns. Raj et al. \cite{raj} presented the use of granular computing approaches for finding community structures within social networks, where nodes are classified based on their similarities and interaction behaviors. Li et al. \cite{lisenti} proposed a granular computing paradigm for sentiment analysis in social networks.
\par

Granular computing incorporates various computational intelligence techniques to solve complex problems involving uncertainty and imprecision. Rough set theory (RST) \cite{lin, yao} is a granular computing technique with applications in various fields, including decision studies \cite{liu}, information sciences \cite{zhang}, and data mining \cite{liao}. The main concept of RST is the indiscernibility relation which is strongly related to data granularity. This relation is defined as an equivalence relation, meaning it satisfies three properties: reflexivity, symmetry, and transitivity. Equivalence relations naturally lead to the partitioning of data into disjoint sets. In the context of networks, these sets can represent communities or clusters of nodes with similar properties or behavior. Gupta and Kumar \cite{Gupta} proposed a rough set-based approach for community detection in networks. Dutta et al. \cite{dutta} introduced an algorithm based on RST for attribute selection and implemented it in the context of spam classification within online social networks. RST helps analyze networks by identifying important nodes that provide information about the overall structure of the network.
\par
Graph theory plays an important role in the study of networks by providing a rich set of tools and techniques for modeling, analyzing, visualizing, predicting, and optimizing complex systems \cite{west}. In practical scenarios, graphs with irregular structures are commonly found in various domains including social networks \cite{raj}, and protein networks \cite{pro}. Numerous real-world applications require graph analysis, which includes tasks like graph classification \cite{chenz}, node classification \cite{zhux}, and link prediction \cite{daud}. Prabhu et al. \cite{prabhu} determined the fault-tolerant metric dimension (FTMD) for butterfly, benes, and silicate networks, highlighting that twin vertices are essential in all fault-tolerant resolving sets. Gargantini et al. \cite{garg} analyzed how certain graph operations, like adding twins or simplicial vertices, impact the rotation graph structure of a given graph. Honkala et al. \cite{honkala} explored the ensemble of optimal identifying codes in twin-free graphs, analyzing how unique vertex identifiers can be constructed with minimal codes, advancing the understanding of graph-based cryptographic applications and efficient network labeling.
\par
The integration of RST and graph theory provides a powerful framework for analyzing and understanding complex systems with applications in various fields. RST can be applied to analyze the structure of graphs and helps in performing granulation of vertices using indiscernibility relation. This involves partitioning the vertex set into indiscernibility classes such that objects within the same class have similar properties. Numerous research efforts have been taken, particularly focusing on undirected graphs. Stell \cite{stell1, stell} was the first to apply granular computing and RST concepts to undirected graphs and proposed a novel approach for analyzing the granularity of hypergraphs \cite{stell2}.
Chiaselotti et al. \cite{Chiaselotti1, Chiaselotti2, Chiaselotti3, Chiaselotti4} investigated granular computing in the framework of graphs and digraphs, utilizing the adjacency matrix as an information table. They also introduced the concept of the metric information table for simple undirected graphs in \cite{Chiaselotti5}. Xu et al. \cite{xu} investigated the connections between generalized RST and graph theory using the mutual representation of binary relations. Javaid et al. \cite{Javaid} studied graphs by using the concept of orbits and RST. Arshad et al. \cite{hibba} studied the zero-divisor graphs of rings of integers modulo $n$ using the concept of granular computing. Fatima et al. \cite{fatima} applied the concept of RST to finite dimensional vector spaces.
\par
Distance metrics in graphs or networks measure the distance between pairs of nodes, providing information about the structural features of graphs. Metrics such as shortest path lengths, commute times, and resistance distances, are important for analyzing the structure and dynamics of a network. These metrics are critical for determining centrality measures, which evaluate the relevance of nodes in a network. The number of shortest paths passing through a node determines betweenness centrality, which helps to identify influential nodes \cite{Freeman}. Closeness centrality measures how close a node is to all others, highlighting those that can quickly interact with the network \cite{Sabidussi}. Newman \cite{newman} introduced techniques of modularity optimization with distance metrics for finding clusters of closely connected nodes in networks.
The metric dimension and related distance-based parameters of networks play a significant role in various real-life applications, including image processing and pattern recognition. The concept of the metric dimension of a graph was first introduced by Slater \cite{Slater} and then independently by Harary and Melter \cite{Harary}. Harary and Melter \cite{Harary} used the concepts of location number and locating set, while Slater \cite{Slater} utilized the terminology of resolving set and metric dimension. Throughout this paper, we will adopt Slater's terms, specifically metric dimension and resolving set, whenever referring to these concepts. In  \cite{Harary}, it is stated that finding the metric dimension of a graph is NP-complete problem. Several algorithms have been developed to compute the metric dimension and resolving sets of a graph. These algorithms can be exact, approximate, or heuristic in nature. Exact algorithms aim to find the exact metric dimension, but they are often computationally expensive, especially for large graphs due to the NP-hardness of the problem. Approximation algorithms provide solutions that are guaranteed to be close to the optimal metric dimension. Heuristic algorithms offer quick, but not necessarily optimal, solutions to the problem.
\par
Granular computing provides a systematic approach to simplify and manage complexity in networks by breaking them down into smaller, more manageable units. In this paper, we study metric-based granular computing techniques for network analysis. We introduce an information system using the distance between vertices and define an indiscernibility relation on the vertex set $V$ with respect to a subset $\mathbb{A} \subseteq V$. We study indiscernibility partitions induced by this indiscernibility relation and establish that reducts and resolving sets are equivalent. Using this equivalence, we introduce a novel method for computing all the minimal resolving sets in networks. We also present two different algorithms for finding resolving sets in networks. The first algorithm is based on the definition of reduct, and identifies different subsets that can serve as resolving sets. The second algorithm uses the discernibility function derived from the discernibility matrix to compute all the minimal resolving sets. We lay foundations for this study for simple undirected graphs and then apply these concepts to two families of zero-divisor graphs associated with $\mathbb{Z}_n$ and $\prod_{i=1}^{k}\mathbb{Z}_2$, representing a family of graphs with twin vertices and a family of twin-free graphs respectively to establish the effectiveness of the proposed techniques in analyzing network structure and properties.
\par
The paper is structured as follows. In section 2, we recall some basic terminology and concepts. Moreover, we associate a political network with zero-divisor graphs and relate it to the problem of finding important nodes in the network. In section 3, we define an indiscernibility relation on the vertex set $V$ with respect to a subset $\mathbb{A} \subseteq V$ using the concept of distance and investigate indiscernibility partitions formed by different vertex subsets.  We also analyze how various vertex subsets contribute to reducts. Moreover, we study the granulation of networks associated with commutative rings. We also study the relationship between positive region, dependency measures, and partial order relation. In section 4, we present the distance-based discernibility matrix to study the properties of networks and show that this matrix can be used to compute all the possible resolving sets. Furthermore, we examine the structure of essential sets and observe that these sets correspond to all the minimal entries of the distance-based discernibility matrix. In Section 5, we illustrate the application of the proposed concepts, and in the final section, we present the conclusion of the paper.
\smallskip
\section{Preliminaries}
This section provides a review of some fundamental concepts, which will be useful in subsequent sections.
\subsection{Information System and Set Partitions}
Information system serves as a formal structure to represent objects, their attributes and the values of these attributes corresponding to each object.
An information system \( \mathcal{I} \) is defined as a quadruple \( (U, Att, f, Val) \), where \( U \) represents a non-empty finite set of objects, known as the universe, \( Att \) is a non-empty set of attributes, and \( f: U \times Att \to Val \) is an information mapping function, with \( Val \) being the set of possible values. An information system \( \mathcal{I} \) is referred to as a Boolean information system if its set of values is \( Val = \{0,1\} \).

\par
Information systems are also known as information tables, data tables, and attribute-value systems. By examining the attribute-value pairs for each object, we can determine which objects are indistinguishable with respect to a given set of attributes $\mathbb{A}\subseteq Att$. The indiscernibility relation between the objects $v_i$ and $v_j$ of $U$ is an equivalence relation defined as \cite{ind}:
\begin{center}
$v_{i}\equiv_{\mathbb{A}} v_{j}\Leftrightarrow f(v_{i}, a)=f(v_{j}, a)\, \forall \, a\in \mathbb{A}$.
\end{center}
This $\mathbb{A}$-indiscernibility relation divides the object set $U$ into disjoint sets such that it induces a partition $\pi_{\mathbb{A}}$ on $U$. We denote by \( [v_i]_{\mathbb{A}} \) the granule of \( v_i \) under the equivalence relation \( \equiv_{\mathbb{A}} \). Any change in the attribute set $\mathbb{A}$ leads to a corresponding change in the composition of the information granules. We say $\mathbb{G}_{\mathbb{A}}(U)=(U, \mathbb{A},\pi_{\mathbb{A}}(U))$ is the granular referencing system of $U$ induced by $\mathbb{A}$.
\par
In RST, each subset $X\subseteq U$ is associated with two subsets of $U$ as lower and upper approximations. Pawlak \cite{ind} defined the lower and upper approximations of $X$ as:
\begin{center}
$\mathbb{L}_{\mathbb{A}}(X)= \{v_{i} \in U: [v_i]_{\mathbb{A}} \subseteq X\}=\bigcup\{[v_i]_{\mathbb{A}}: [v_i]_{\mathbb{A}} \subseteq X\}$\\
$\mathbb{U}_{\mathbb{A}}(X)= \{v_{i} \in U: [v_i]_{\mathbb{A}}\cap X\neq\emptyset\}=\bigcup\{[v_i]_{\mathbb{A}}: [v_i]_{\mathbb{A}}\cap X\neq\emptyset\}$.
\end{center}
A subset $X\subseteq U$ is referred to as $\mathbb{A}$-definable/$\mathbb{A}$-exact if $\mathbb{L}_{\mathbb{A}}(X)=\mathbb{U}_{\mathbb{A}}(X)$; otherwise the set is known as $\mathbb{A}$-rough \cite{ind}.
\par
An important aspect of RST involves identifying and removing redundant or irrelevant attributes while preserving important information. In this process, the concept of reduct plays a significant role. A subset $\mathbb{A} \subseteq Att$ is said to be a reduct of an information system $\mathcal{I}$ if $\pi_{\mathbb{A}}=\pi_{Att}$ and $\pi_{\mathbb{A}\setminus\{v_{i}\}}\neq \pi_{Att}$ for all $v_{i} \in \mathbb{A}$ \cite{ind}.
An information system can possess multiple reducts, where \( RED \) denotes the set of all reducts associated with \( \mathcal{I} \). The core is defined as the intersection of all reducts.
For any pair $v_i, v_j \in U$, the discernibility matrix is a $|U|\times|U|$ matrix whose entries $(i, j)$ are defined as $\Delta_{\mathcal{I}}(v_{i}, v_{j}) :=\{a\in Att: f(v_{i}, a) \neq f(v_{j}, a)\}$.

\subsection{Networks}
A network is a collection of objects, often called nodes or vertices, that are connected by relationships, referred to as edges. Networks appears in different contexts such as social relationships in social networks, hyperlinks between web pages, transportation routes between cities, and molecular interactions in biological systems. Graphs are visual representations of these types of networks. Graph theory provides a set of mathematical tools and parameters for studying and analyzing networks. Assuming the reader has prior knowledge of the terminology, we provide a brief overview of the basic notations and refer to \cite{west} for more detailed explanations.
\par
Formally, network (or graph) $G$ is defined as an ordered pair consisting of a vertex set $V$ and edge set $E$, written as $G=(V, E)$ where each edge connects two vertices. The degree of a vertex is the number of edges that are incident to it. The total number of vertices in $G$ is referred to as the order of $G$, and the total number of edges in $G$ is known as its size. The neighbourhood of a vertex $v_{i}$ is defined as $N_{G}(v_{i}) = \{ v_{j} \in V : (v_{i}, v_{j})\in E\}$. For a subset $X$ of $V$, the neighbourhood of $X$ is defined as $N_{G}(X) = \bigcup_{v_{i} \in X} N_{G}(v_{i})$. Two vertices in $G$ are said to be true twins if they share the same neighborhood and are also connected by an edge. The vertices are called false twins if they share the same neighborhood but are not directly connected.
\par
The shortest path length, or the minimum number of edges, between two vertices in a network is referred to as their distance and is denoted by $d(v_{i}, v_{j})$. The diameter of a connected graph is described as the maximum distance between any two vertices. For $v_{i}, v_{j} \in V$, if $d(v_{i}, u) \neq d(v_{j}, u)$, we say that $u \in V$ resolves or distinguishes $v_{i}$ and $ v_{j}$ in $V$. A subset $A$ of vertex set $V$ is said to be a resolving set of $G$ if there exists at least one vertex $u \in A$ that resolves $v_{i}$ and $v_{j}$ for each pair $v_{i}, v_{j} \in V$ \cite{Slater}. The metric dimension of \( G \), denoted by \( \dim(G) \), is the cardinality of a minimal resolving set with the smallest number of vertices, known as a metric basis of \( G \). An upper basis of \( G \) is a minimal resolving set that contains the maximum number of vertices. Its cardinality is referred to as the upper dimension of \( G \), represented by \( \dim^{+}(G) \) \cite{redmond}.
\par
We apply our methodology to two families of zero-divisor graphs associated with commutative rings.
For a ring $R$, Beck \cite{Beck} proposed the concept of zero-divisor graphs, in which the vertices correspond to the elements of a ring. Later, Livingston \cite{Livingston} proposed a new definition of zero-divisor graphs, where the vertices are the zero-divisors of the ring. For the ring $\mathbb{Z}_{n}$, zero-divisor graph is $\Gamma(\mathbb{Z}_{n})$, where the vertices are the zero-divisors of $\mathbb{Z}_{n}$, denoted by $Z(\mathbb{Z}_{n})$. For any two vertices $v_{i}, v_{j} \in Z(\mathbb{Z}_{n})$, there is an edge between $v_{i}$ and $v_{j}$ if $v_{i}\cdot v_{j}=0$. According to  Rather et. al \cite{Rather}, the ring $\mathbb{Z}_{n}$ contains $n-\phi(n)-1$ zero-divisors. The zero-divisor graph of the ring $\prod_{i=1}^{k}\mathbb{Z}_2$, where $k \geq 2$, is denoted by $\Gamma(\prod_{i=1}^{k}\mathbb{Z}_2)$ and has a vertex set consisting of all zero-divisors of $\prod_{i=1}^{k}\mathbb{Z}_2$, (i.e., $Z(\prod_{i=1}^{k}\mathbb{Z}_2) = \prod_{i=1}^{k}\mathbb{Z}_2 \setminus \{(0, 0, \cdots, 0), (1, 1, \cdots, 1)$)\}. The order of this graph is $2^k - 2$. The graphs associated with the ring $\mathbb{Z}_{n}$  are classified as graphs with twins, while those associated with $\prod_{i=1}^{k}\mathbb{Z}_2$ form twin-free graphs. Throughout this paper, we assume $R$ to be a commutative ring and $V$ be a vertex set of $G$ associated with $R$, unless stated otherwise. Moreover, the terms "network" and "graph" are used interchangeably, indicating the same mathematical structure and concept.
\par
We now illustrate the concept of resolving set using a specific example from social network analysis.
\subsection{Political Network Modeled by Graphs Associated to Algebraic Structure}\label{pn}
Social interaction encompasses both positive and negative relationships. People establish connections to express friendship, support, or approval, while also forming links to signify disapproval, disagreement, or distrust. Research on social networks has predominantly focused on positive relationships, with negative relationships receiving comparatively less attention \cite{zare}.  In this example, we study a network in which relationships are negative, which indicates that the individuals are in opposition.
\par
Consider a political network modeled as a graph $G$ where the vertex set consists of individuals in the political party, and there is an edge between two individuals if the corresponding individuals hold opposing views on certain issues. Let $V=\{v_1, v_2, \cdots, v_n\}$ be a set of individuals (or agents) expressing views on a particular issue. Each agent $v_i$ holds a view that contrasts with the views of other agents, referred to as opponents.
Now, suppose there are $k$ distinct issues, each of which can be supported or opposed. Each individual’s opinion can be represented as a binary string of length $k$, where 1 denotes support and 0 denotes opposition. If an individual supports issue 1 and issue 4 and opposes all other issues then that individual is represented by a vector of length $k$ as $(1, 0, 0, 1, 0, \cdots, 0)$.
\par
We model this situation using the concept of zero-divisors, where we interpret an individual’s binary opinion vector as an element in the ring $\prod_{i=1}^{k}\mathbb{Z}_{2}$, the direct product of $k$ copies of the ring $\mathbb{Z}_2$. In this case, the individuals are modeled as elements in this ring, and two individuals are connected by an edge if the product of their opinion vectors (interpreted as elements of $\prod_{i=1}^{k}\mathbb{Z}_{2}$) is zero. For example, consider two individuals, $v_1$ and $v_2$, where $v_1$ holds the opinion vector $(1, 0, 0)$ and $v_2$ holds $(0, 1, 1)$. The product of these two vectors is $(0, 0, 0)$, indicating that their views on certain issues (in this case, all issues) are mutually exclusive. Thus, they are connected by an edge in the graph, and they behave as zero-divisors.

\par
We consider the zero-divisor graph of $\prod_{i=1}^{3}\mathbb{Z}_{2}$ where each vertex in the zero-divisor graph represents a combination of binary values $(a, b, c)$, corresponding to opinions on three political issues: $a$ for taxation, $b$ for healthcare, and $c$ for education. An edges between vertices signify complete opposition or disagreement on all issues. For example, a vertex labeled as $(1, 0, 0)$ might represent an individual who is in support of taxation, but opposes healthcare and education.
Suppose there are six individuals, each with opinions on three different issues (taxation, healthcare, education):\\
David: Supports taxation, opposes healthcare, supports education.\\
Sarah: Opposes taxation, supports healthcare, opposes education. \\
Emily: Supports taxation, supports healthcare, opposes education. \\
Maria: Opposes taxation, opposes healthcare, supports education. \\
John: Supports taxation, opposes healthcare, opposes education. \\
Jessica: Opposes taxation, supports healthcare, supports education.\\
The political network is illustrated in Figure \ref{fig11}.
\begin{figure}[h!]
  \centering
\includegraphics[width=5cm]{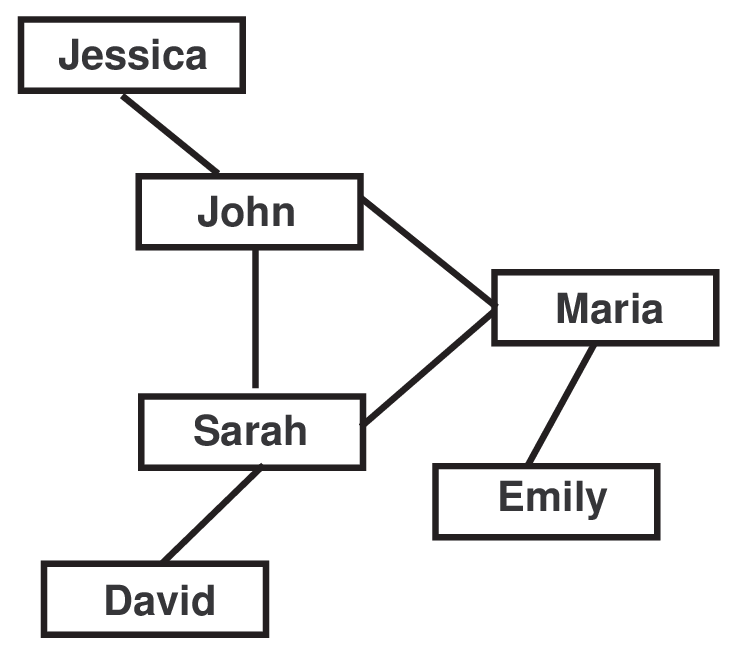}
   \caption{Political Network}%
   \label{fig11}%
\end{figure}
\par
In network analysis, one of the fundamental tasks is identifying the smallest subsets of nodes in a network that uniquely distinguish every other node in the network. This is the problem of finding resolving sets in a network. The task of finding all resolving sets of a network is generally challenging because it involves an exhaustive search over all possible subsets of vertices, and the number of such subsets can be exponential large in size. Therefore, we use RST to identify all resolving sets. A resolving set in the political network $G$ is defined as a subset of individuals whose perspectives or opinions uniquely differentiate every other individual in the network based on their opposition relationships. For instance, a possible resolving set might be $\{John, Maria\}$. When compared to an agent like Emily, whose opinion vector is
$(1,1,0)$ (supports taxation and healthcare, opposes education), John and Maria’s opposing views allow them to uniquely identify Emily. John agrees with Emily on taxation but differs on healthcare, while Maria differs on taxation and education. This ability to differentiate based on opposition holds for any other agent in the network, ensuring that John and Maria's collective perspectives can distinguish all individuals based on their views on the three issues.

\par
By using this framework, we can model social and political networks as graphs where opposition relationships are represented through zero-divisors in a ring structure, highlighting how contrasting opinions lead to disjoint or opposing positions within the network.
\smallskip
\section{Granular Computing in Networks}
In this section, we represent finite simple undirected graphs (or networks) as an information table by using the distance between vertices as an information map
to obtain the indiscernibility partitions. In this context, the network is considered as a system, where the nodes represent objects, and the edges denote relationships or interactions between these objects. We define an indiscernibility relation on the vertex set $V$ of a simple undirected graph with respect to $\mathbb{A}\subseteq V$, and study the resulting indiscernibility partitions on $V$. We also establish that reducts and minimal resolving sets are equivalent. Moreover, we study the granulation of networks modeled by zero-divisor graphs.
\par
We start by defining the concept of an information table for simple undirected graphs, which serves as the foundation for this study.
\begin{defn}
An information table $\mathcal{I}(\mathcal{G})$ of a graph $\mathcal{G}$ is a quadruple $(V, \mathbb{A}, \mathcal{F}, Val)$, where  $V=\{v_1, v_2, \cdots v_n\}$ is a vertex set of $\mathcal{G}$ and $\mathbb{A}\subseteq V$, $Val = \{0, 1, 2, \cdots , diam(\mathcal{G})\}$ and information map $\mathcal{F}: V\times V\rightarrow Val$ is defined as:
\[
 \mathcal{F}(v_{i}, v_{j}) = \left.
  \begin{cases}
   d(v_{i}, v_{j}), & \text{if }  v_{i}\neq v_{j}  \\
     0, & \text{if }  v_{i}= v_{j}.
    \end{cases}
  \right.
\]
\end{defn}
From now, we assume $\mathcal{I}= \mathcal{I}(\mathcal{G})$ unless mentioned otherwise. In a traditional information table, each row label represents an object, and each column label represents an attribute. The values in the table represent the attribute values associated with each object. In the context of graph analysis, the information table can be represented as a matrix, where both the rows and columns correspond to vertices, and the entries represent relationships between pairs of objects. We consider a distance matrix as an information table where rows and columns represent objects, and the entries represent distances between pairs of objects.
\par
Two vertices \( v_i \) and \( v_j \) are considered indiscernible with respect to a vertex \( u \) if \( d(v_i, u) = d(v_j, u) \). More generally, \( v_i \) and \( v_j \) are indiscernible with respect to \( \mathbb{A} \) if this condition holds for all \( u \in \mathbb{A} \). The vector of distances $(d(_{i},u_{1}), d(v_{i}, u_{2}), \cdots, d(v_{i}, u_n))$ is the ordered tuple and is denoted by $\gamma(v_{i}|\mathbb{A})$. If $\mathbb{A}=V$, then the row of the information table corresponding to $v_i$ is same as $\gamma(v_{i}|V)$. Note that, the $i^{th}$ vertex in $\mathbb{A}$ has $0$ in its $i^{th}$ coordinate and all other coordinates are non-zero. As a result, the vertices of $\mathbb{A}$ must have different representations.
Two vertices $v_{i}$ and $v_j$ in $V$ are $\mathbb{A}$-indiscernible denoted as $v_{i} \equiv_{\mathbb{A}} v_j$ if and only if $\gamma(v_{i}|\mathbb{A})=\gamma(v_j|\mathbb{A})$. Equivalently, $v_{i} \equiv_{\mathbb{A}} v_j\Longleftrightarrow \mathcal{F}(v_{i}, a)=\mathcal{F}(v_j, a)$\, for all $a\in \mathbb{A}$.
The set of all the vertices indiscernible to a vertex $v_{i}$ with respect to $\equiv_{\mathbb{A}}$ is defined as:
\begin{equation}\label{eq2}
C_{\mathbb{A}}(v_{i})=\{v_j \in V:\mathcal{F}(v_j, a)=\mathcal{F}(v_{i}, a)\,\, \forall\, a \in \mathbb{A} \}.
\end{equation}
that is equivalent to
\begin{equation}\label{eq3}
  v_{i} \equiv_{\mathbb{A}} v_j\Longleftrightarrow C_{\mathbb{A}}(v_{i})=C_{\mathbb{A}}(v_j).
\end{equation}
When $\mathbb{A}=V$, then $C_{\mathbb{A}}(v_{i})=\{v_i\}$. Each equivalence class represents an information granule, consisting of vertices that are indistinguishable from one another. That is if $C_i\subseteq V$ such that
$ C_i=C_{\mathbb{A}}(v)$, for some $v \in V$, we say that $C_i$ is an $\mathbb{A}-granule$ of $\mathcal{I}$.
As $C_1, C_2, \cdots, C_s$ are the distinct granules of $\mathcal{I}$, we use the notation $\pi_{\mathbb{A}}=C_1| C_2| \cdots|C_s$ to represent the indiscernibility partition of the vertex set $V$.
\par
Two vertices $v_{i}$ and $v_j$ are distance-similar if $d(v_{i},w) = d(v_j,w)$ $\forall \, w \in V \setminus \{v_{i}, v_j\}$. The set of all distance-similar vertices is denoted by $\mathcal{B}$. If a graph consists of $k$ distance-similar groups, we say $\pi$ is a distance-similar partition and $\pi=\mathcal{B}_{1}|\mathcal{B}_{2}|\cdots|\mathcal{B}_{k}$. The indiscernibility relation $\equiv_{\mathbb{A}}$ is based on the distance between vertices with respect to a subset $\mathbb{A}$, the distance-similarity condition considers the distances to all other vertices in $V$.
\begin{Example}
Consider a customer purchase network in Figure \ref{figc}. This network is a representation of the relationships between customers based on the products they purchase, where nodes represent customers and edges represent different purchases. We have seven customers $c_1, c_2, c_3, c_4, c_5, c_6, c_7$ and four products milk, bread, cheese and sugar. In this scenario, customers $c_1$ and $c_6$ have purchased milk and cheese, $c_3$ and $c_5$ have purchased milk, $c_2$ and $c_7$ have purchased bread and cheese, and $c_4$ has purchased bread and sugar.
\begin{figure}[h!]
  \centering
\includegraphics[width=5cm]{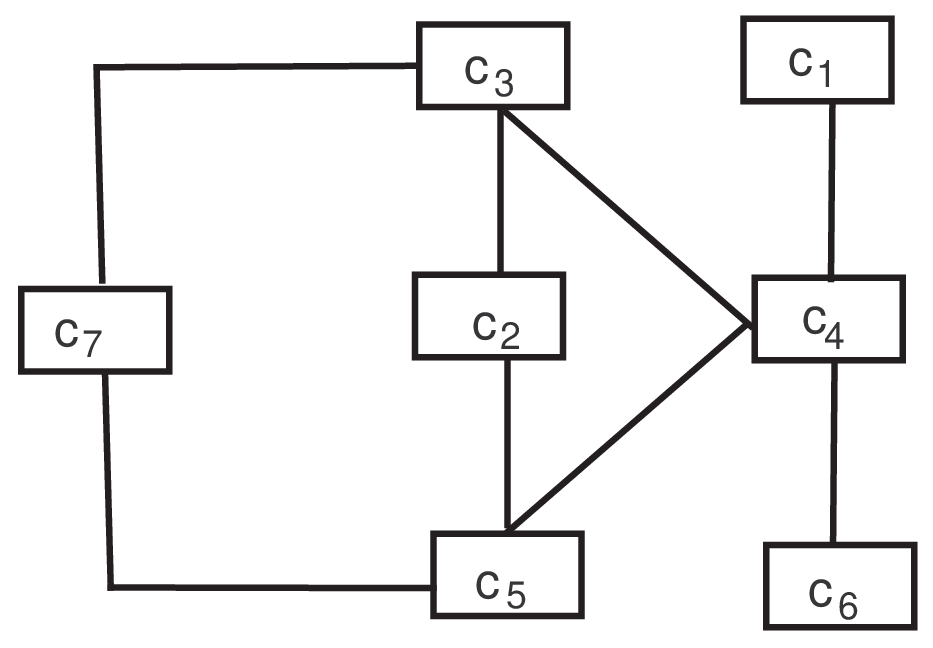}
   \caption{Customer Purchase Network}%
   \label{figc}%
\end{figure}
When analyzing this network, it becomes evident from the distances between the vertices that vertices representing $c_1$ and $c_6$ are equivalent, $c_2$ and $c_7$ are equivalent, $c_3$ and $c_5$ are equivalent. In RST, such equivalence indicates that these two customers cannot be distinguished from each other. Therefore the distance-based partition of the vertex set is $c_1, c_6|c_2, c_7|c_3, c_5|c_4$.
\begin{table}[h!]
\scriptsize
\centering
\caption{Information System for Customer Purchase Network}
\begin{tabular}{c|ccccccc}
  \hline
   & $c_1$ & $c_2$ & $c_3$ &$ c_4$ & $c_5$ & $c_6$ & $c_7$\\
\hline
  $c_1$ & 0 & 3 & 2 & 1 & 2 & 2 & 3\\
  $c_2$ & 3 &0 & 1 & 2& 1& 3 &2  \\
  $c_3$ & 2& 1 & 0 & 1& 2 & 2& 1 \\
 $c_4$ & 1 & 2 & 1 & 0 & 1& 1& 2\\
  $c_5$  & 2& 1 & 2 & 1& 0 & 2& 1 \\
$c_6$ & 2 &3 & 2 & 1& 2& 0 & 3  \\
$c_7$& 3& 2 & 1 & 2 & 1 & 3 & 0\\
\hline
\end{tabular}

\label{Table1}
\end{table}
\par
Rows of this information table give representation of vertices with respect to $V$ and columns represent attributes. If we consider $\mathbb{A}\subseteq V$, the representation of vertices with respect to the attributes in $\mathbb{A}$ is the collection of those entries in the information system that correspond to $\mathbb{A}$. For instance, the vertices \( c_2 \) and \( c_7 \) are indiscernible with respect to \( \mathbb{A} = \{c_1\} \), but they do not remain indiscernible when considered with respect to the entire set \( V \).
\end{Example}
In social networks, the indiscernibility relation can partially order nodes (individuals) based on shared attributes, such as influence or authority. A partial order relation $\preceq$ is a binary relation which is reflexive, antisymmetric and transitive. The relationship between two sets $\mathbb{A}, \mathbb{A'}\subseteq V$, establishes a partial order relation $\preceq$ among the corresponding indiscernibility partitions as; $\mathbb{A} \subseteq \mathbb{A'}\Leftrightarrow \pi_{\mathbb{A'}}\preceq \pi_{\mathbb{A}}$. The partial ordering of two set partitions is primarily determined by the inclusion relationship among their granules.
For every $v_{i} \in V$, we define $\preceq$ as $\pi_{\mathbb{A}}\preceq \pi_{\mathbb{A'}}\Leftrightarrow C_{\mathbb{A}}(v_{i})\subseteq C_{\mathbb{A'}}(v_{i})$. We say that $\pi_{\mathbb{A}}$ is finer than $\pi_{\mathbb{A'}}$ and $\pi_{\mathbb{A'}}$ is coarser than $\pi_{\mathbb{A}}$ if $\pi_{\mathbb{A}}\preceq \pi_{\mathbb{A'}}$.
Note that, if  $\mathbb{A}=\emptyset$, partition induced by $\mathbb{A}$ is $\pi_{\mathbb{A}}= v_1, v_2, \cdots, v_n$.
\begin{Remark}
For $\mathcal{I}$, the partition induced by $V$ represents the finest partition, and the partition induced by the empty set is considered as the coarsest partition.
\end{Remark}
\par
Let $\Pi_{V}=\{\pi_{\mathbb{A}}: \mathbb{A} \subseteq V\}$, denote the collection of all possible distinct partitions of $V$. The pair ${\Pi_{ind(V)}}:= (\Pi_{V}, \preceq)$  is referred to as the indiscernibility partition lattice. The order within a partition lattice illustrates how partitions refine or coarsen each other. The meet operation identifies the largest partition $\pi_{\mathbb{A}}\wedge \pi_{\mathbb{A'}}$, that refines both $\pi_{\mathbb{A}}$ and $\pi_{\mathbb{A'}}$, while join represents the smallest partition $\pi_{\mathbb{A}}\vee \pi_{\mathbb{A'}}$, that is a coarsening of both $\pi_{\mathbb{A}}$ and $\pi_{\mathbb{A'}}$. A partially ordered set (poset) is termed as complete lattice if it possesses both join and meet operations.
\par
Two or more lattices are said to be isomorphic if there exists a bijective mapping between their elements that preserves the order structure. In the context of distance-similar classes, the indiscernibility partition lattices generated by any two subsets within a distance-similar class $\mathcal{B}_{i}$ are isomorphic, meaning they share the same structure in terms of the ordering of partitions. Formally, for $\mathcal{I}$, let $\mathbb{A}, \mathbb{B}\subseteq \mathcal{B}_{i}$ such that $|\mathbb{A}|=|\mathbb{B}|$ then ${\Pi_{ind(\mathbb{A})}}$ is isomorphic to ${\Pi_{ind(\mathbb{B})}}$.
\par
The following proposition establishes the relationship between the indiscernibility relation and distances between vertices of the graph $\mathcal{G}$.
\begin{Proposition}\label{prop0}
For $\mathcal{I}$, let $\mathbb{A}\subseteq V$ and $ v_{i}, v_j \in V\setminus \mathbb{A}$  then the following statements are equivalent.\\
(i) $v_{i}\equiv_{\mathbb{A}}v_j$\\
(ii) $d(v_{i}, u)=d(v_j, u)$ $\forall u\in \mathbb{A}$\\
(iii) $\gamma(v_i|\mathbb{A})=\gamma(v_j|\mathbb{A})$
\end{Proposition}
\par
Note that, if $v_i, v_j \in \mathbb{A}$, then $v_i \not\equiv_{\mathbb{A}} v_j$. As $d(v_i, v_i) = 0$ and $d(v_j, v_j) = 0$, it follows that $d(v_i, u) \neq d(v_j, u)$ for all $u \in \mathbb{A}$ which implies that  $\gamma(v_i | \mathbb{A}) \neq \gamma(v_j | \mathbb{A})$. Proposition \ref{prop0} provides an interpretation for the indiscernibility relation $\equiv_{\mathbb{A}}$. The equivalence relation $\equiv_{\mathbb{A}}$ is described as a type of distance-based similarity relation concerning the vertex subset $\mathbb{A}$. For example, let’s consider two customers from customer purchase network in Figure \ref{figc}: $c_{1}$ and $c_{5}$. If their distances to a third customer, say $c_{6}$, are equal this implies that $c_{1}$ and $c_5$ have a similar relationship with $c_{6}$ in terms of the products being purchased.

\par
The following remark establishes that if two vertices are \( \mathbb{A} \)-indiscernible, then they remain indiscernible with respect to every subset of \( \mathbb{A} \).
\begin{Remark}
Suppose $\mathbb{A}, \mathbb{B}\subseteq V$, if $v_{i}\equiv_{\mathbb{A}} v_j$ then $v_{i}\equiv_{\mathbb{B}}v_j$ for all $\mathbb{B}\subseteq \mathbb{A}$.
\end{Remark}
\par
It is interesting to note that if $w_1, w_2\in \mathbb{A}$ be two vertices such that $d(v, w_1)=d(v, w_2)$ for all $v \in V$. Then, removing $w_2$ from $\mathbb{A}$ does not affect the partition induced by $\mathbb{A}$.
\par
Let \( G= (V, E) \) be a graph. An automorphism of \( G \) is a bijective function \( \phi: V \to V \) such that for all \( u, v \in V \), \( \{u, v\}\in E\) if and only if \( \{\phi(u)\phi(v)\} \in E \). The set of all such automorphisms forms a group under function composition, denoted by \( Aut(G) \).  Consider a vertex \( v_i \in V \). The vertex set can be partitioned based on its distance from $v_i$ as \( \pi_{v_i} = V_0 \mid V_1 \mid \dots \mid V_k \) where \( V_i = \{ v_j \in V \mid d(v_i, v_j) = i \} \) and $i\in \{1, 2, \cdots k\}$. If $ \phi \in \text{Aut}(\mathcal{\mathcal{G}}) $, then $\pi_{\phi(v_i)} = \phi(V_0) \mid \phi(V_1) \mid \dots \mid \phi(V_k),$ where \( \phi(V_i) = \{ \phi(v_j) \mid v_j \in V_i \} \). In general, we have the following result which states that the partition of the vertex set $V$ induced by a subset $\mathbb{A}$ remains invariant under the action of an automorphism of the graph $\mathcal{G}$.
\begin{Proposition}\label{auto}
For $\mathbb{A} \subseteq V$, the partition $\pi_\mathbb{A}$ is preserved under the action of any automorphism $ \phi \in Aut(\mathcal{G})$.
\end{Proposition}
\begin{prf}
Since $\phi$ is an isometry, thus, for any $v_1, v_2 \in V$ and $a \in \mathbb{A}$, we have $d(v_1, a) = d(v_2, a) \implies d(\phi(v_1), \phi(a)) = d(\phi(v_2), \phi(a))$. This shows that if $v_1\equiv_{\mathbb{A}}v_2$ then $\phi(v_1)\equiv_{\phi(\mathbb{A})}\phi(v_2)$. Let $\pi_\mathbb{A} = C_1|C_2|\cdots |C_k$, and for any block $C_i$, consider $\phi(C_i) = \{\phi(v_1) \mid v_1 \in C_i\}$.
We claim that $\phi(C_i)$ is a block of $\pi_{\phi(\mathbb{A})}$. Suppose $v_1, v_2\in C_i$, then $v_1 \equiv_{\mathbb{A}} v_2$, so $d(v_1, a) = d(v_2, a)$ for all $a \in \mathbb{A}$. By the action of $\phi$, $d(\phi(v_1), \phi(a)) = d(\phi(v_2), \phi(a))$ for all $a \in \mathbb{A}$, meaning $\phi(v_1) \equiv_{\phi(\mathbb{A})} \phi(v_2)$.
Therefore, $\phi(v_1), \phi(v_2) \in \phi(C_i)$, and $\phi(C_i)$ forms a block of $\pi_{\phi(\mathbb{A})}$.
Since $\phi$ is bijective, every vertex in $V$ belongs to exactly one block $\phi(C_i)$, and no two such blocks overlap. Thus, $\pi_{\phi(\mathbb{A})}$ consists of the blocks $\phi(C_1), \phi(C_2), \ldots, \phi(C_k)$. Finally, we conclude that $\pi_{\phi(\mathbb{A})} = \phi(C_1)|\phi(C_2)|\cdots |\phi(C_k),$
which proves the proposition.
\qed
\end{prf}
\par
The following result establishes that the number of granules in $\pi_\mathbb{A}$ is determined by the cardinality of $\mathbb{A}$.
\begin{Proposition}
For $\mathbb{A}\subseteq V$, we have $1 \leq|\pi_\mathbb{A}| \leq (diam(\mathcal{G})+1)^{|\mathbb{A}|}$.
\end{Proposition}
\begin{prf}
The proof is structured into the following three cases.\\
(i) Suppose $\mathbb{A}\subseteq V$ and $\mathbb{A}=\emptyset$ then all the vertices in $V$ are indiscernible that is $\pi_\mathbb{A}=v_1, v_2,\cdots, v_n$ implies that $|\pi_\mathbb{A}|=1$.\\
(ii) For any nonempty \( \mathbb{A} \subseteq V \), $|\pi_\mathbb{A}|\geq |\mathbb{A}|+1$ and each block of $\pi_\mathbb{A}$ has elements with the same representations. The distance representation of a vertex $v$ relative to \( \mathbb{A} \) is $\gamma(v| \mathbb{A}) = (d(v, a_1), d(v, a_2), \dots, d(v, a_k)),$ where $d(v, a_i) \in \{0, 1, \dots, diam(\mathcal{G})\}$ for all $i\in \{1, 2, \cdots, k\}$. Hence there are at most $(diam(\mathcal{G})+ 1)^{|\mathbb{A}|}$ distance representations.
\qed
\end{prf}
\par
If the vertex set of the graph is partitioned into distance-similar classes and we consider a subset $\mathbb{A}$ with one element from each class, then the induced partition $\pi_\mathbb{A}$ is given in the following result.
\begin{Proposition}
For $\mathcal{I}$, let $\mathbb{A}=\{v_1, v_2, \cdots v_s\}$ and $|\mathbb{A}\cap\mathcal{B}_i|=1$ for all $i\in \{1,2,\cdots, k\}$ then $\pi_{\mathbb{A}}=v_{1}| \cdots|v_{s}| \mathcal{B}_{1}\setminus \mathbb{A}|\mathcal{B}_{2}\setminus \mathbb{A}|\cdots |\mathcal{B}_{k}\setminus \mathbb{A}$.
\end{Proposition}
\begin{prf}
Suppose $|\mathbb{A}\cap\mathcal{\mathcal{B}}_i|=1$ then for any $v_i, v_j \in \mathbb{A}$, we have $v_{i}\not \equiv_{\mathbb{A}} v_j$. Suppose $v_i, v_j\in \mathcal{B}_{i}\setminus \mathbb{A}$ for all $i=1, 2, \cdots k$ then $d(v_i, w)=d(v_j, w)$ for all $w \in \mathbb{A}$ implies that $v_{i}\equiv_{\mathbb{A}} v_j$. Now let $v_i\in \mathcal{B}_{s}\setminus \mathbb{A}$ and $v_j\in \mathcal{B}_{t}\setminus \mathbb{A}$ then $d(v_i, w)\neq d(v_j, w)$ for all $w \in \mathbb{A}$ implies that $v_{i}\not\equiv_{\mathbb{A}} v_j$.
\qed
\end{prf}
\par
Two different subsets can induce the same or different partitions on the vertex set $V$. We say that two subsets are equivalent  if they produce the same partitions, consider $\mathbb{A}, \mathbb{A'}\subseteq V$, we set
\begin{center}
$\mathbb{A}\approx\mathbb{A'}\Leftrightarrow \pi_{\mathbb{A}}=\pi_{\mathbb{A'}}$.
\end{center}
The equivalence class of $\mathbb{A}$, denoted by $[\mathbb{A}]^{\approx}$ is defined as $[\mathbb{A}]^{\approx}=\{\mathbb{A'}\subseteq V: \mathbb{A}\approx\mathbb{A'}\}$. This leads us to the notions of maximum and minimum partitioners. The maximum partitioner of $\mathbb{A}$, denoted by $Max(\mathbb{A})$, is defined as the union of all elements in $[\mathbb{A}]^{\approx}$ with $Max(\mathbb{A})$ being the largest set in $[\mathbb{A}]^{\approx}$. Similarly, a set $B\in [\mathbb{A}]^{\approx}$ is referred to as minimum partitioner of $\mathbb{A}$, denoted by $Min(\mathbb{A})$, if $\pi_B =\pi_\mathbb{A}$ and for all $B' \subset B$, $\pi_{B'} \neq \pi_B$.
\par
The following result gives the properties of maximum and minimum partitioners of a subset of $V$ and shows the connection between subsets and their maximum partitioners.
\begin{Proposition}\label{max}
For $\mathcal{I}$, let $\mathbb{A}, \mathbb{A}'\subseteq V$ then\\
(i) $\mathbb{A}\approx\mathbb{A}' \Leftrightarrow Max(\mathbb{A})= Max(\mathbb{A'})$.\\
(ii) $\pi_{\mathbb{A}\cup \mathbb{A}'}=\pi_{Max(\mathbb{A})\cup Max(\mathbb{A}')}$.\\
(iii) $Min(\mathbb{A}) \subseteq\mathbb{A}\subseteq Max(\mathbb{A})$.
\end{Proposition}
\begin{prf}
(i) Suppose $\mathbb{A}\approx\mathbb{A}'$ by definition of equivalence class of $\mathbb{A}$, follows that $\mathbb{A}'\in [\mathbb{A}]_{\approx}$ that gives $Max(\mathbb{A})= Max(\mathbb{A}')$. Conversely, if $Max(\mathbb{A})=Max(\mathbb{A}')$ again by definition of $[\mathbb{A}]_{\approx}$ we have $\mathbb{A}\approx Max(\mathbb{A})$ and $\mathbb{A}'\approx Max(\mathbb{A'})$, therefore by transitivity $\mathbb{A}\approx\mathbb{A}'$. \\
(ii) Suppose $\mathbb{A}, \mathbb{A}'\subseteq V$, by definition of maximum partitioner $\mathbb{A}\cup \mathbb{A}'\subseteq Max(\mathbb{A})\cup Max(\mathbb{A}')$ which implies $\pi_{\mathbb{A}\cup \mathbb{A}'}=\pi_{Max(\mathbb{A})\cup Max(\mathbb{A}')}$.\\
(iii) The proof follows from the definitions of the minimum partitioner and maximum partitioner.
\qed
\end{prf}
\par
In the previous result, we identified the different types of subsets that generate the same partition. Now, we turn our focus to identifying subsets of $V$ that induce the same partition as the partition induced by $V$ itself. In other words, we find subsets \( \mathbb{A} \subseteq V \) that belong to $[V]_\approx$. The following result characterizes such subsets that induce the same partition as \( V \).
\begin{Proposition}\label{prop1}
For $\mathcal{I}$, we have the following: \\
(i) If $\pi_{\mathbb{B}}=\pi_{V}$ then $\pi_{\mathbb{A}}=\pi_{V}$ for all $\mathbb{B}\subseteq \mathbb{A}$.\\
(ii) If $|\mathcal{B}_{i}\setminus \mathbb{A}|\leq 1$ for all $i\in \{1,2,\cdots, k\}$ then $\pi_{V}=v_{1}|v_{2}| \cdots |v_n=\pi_{\mathbb{A}}$.
\end{Proposition}
\begin{prf}
(i) Suppose $\pi_{\mathbb{B}}=\pi_{V}$ and let $v_{i}\equiv_{V}v_j$ then by definition of distance $\pi_{V}$ is the finest partition which implies that for any $\mathbb{A}\subseteq V$ we have $v_{i}\equiv_{\mathbb{A}}v_j$ implies that $\pi_{\mathbb{A}}=\pi_{V}$.\\
(ii) Suppose contrary that $|\mathcal{B}_{i}\setminus \mathbb{A}|> 1$. Then there exist at least two
vertices $v_{i}$ and $v_j$ in $\mathcal{B}_{i}$ such that $d(v_{i},w) = d(v_j,w)$ $\forall \, w \in V\setminus \mathcal{B}_{i}$ . This implies that
$v_{i}$ and $v_j$ belong to the same class. This contradicts that $\pi_{V}=v_{1}|v_{2}| \cdots |v_n=\pi_{\mathbb{A}}$.
\qed
\end{prf}
\par
In the previous results, we identified subsets that generate the same partition as the complete vertex set $V$. Now, we are focused on finding the minimal subsets of $V$ that generate the same partition as $V$ itself. This is similar to the metric dimension problem. In the metric dimension problem, the objective is to find the smallest set of vertices $\mathbb{A}$ known as minimum resolving set that uniquely determine the positions of vertices in a graph i.e, $\pi_{\mathbb{A}}=v_1|v_2|\cdots|v_n$. This concept is analogous to the concept of a reduct. Both resolving sets and reducts share the common theme of minimality. They aim to identify the smallest sets that give the same information as the complete attribute set.
\par
In the following result we establish that reduct and minimal resolving sets are equivalent.
\begin{Proposition}
For $\mathcal{G}$, the reduct of $\mathcal{G}$ and minimal resolving sets of $\mathcal{G}$ are equivalent.
\end{Proposition}
\begin{prf}
Suppose $\mathbb{A}$ is a reduct, then $\pi_{\mathbb{A}}=\pi_{V}$, and for all $v_i, v_j \in V$, $v_i \not\equiv_{\mathbb{A}} v_j \Rightarrow \gamma(v_i | \mathbb{A}) \neq \gamma(v_j | \mathbb{A})$. This implies that $\mathbb{A}$ distinguishes all vertices, which establish that $\mathbb{A}$ is a resolving set. Furthermore, if $\mathbb{A}$ is not minimal, then there exist $\mathbb{A}' \subset \mathbb{A}$ such that $\pi_{\mathbb{A}'}=\pi_{V}$, which is a contradiction to the definition of the reduct. Thus, $\mathbb{A}$ is a minimal resolving set.
Conversely, suppose $\mathbb{A}$ is a minimal resolving set, then every $v_i, v_j \in V$ has a unique distance vector relative to $\mathbb{A}$ as $\gamma(v_i | \mathbb{A}) \neq \gamma(v_j | \mathbb{A})$, which implies that each element has unique representation and  $\pi_{\mathbb{A}}=v_1|v_2|\cdots|v_n=\pi_{V}$. Thus, $\mathbb{A}$ satisfy the definition of a reduct. Minimality follows because removing any element from $\mathbb{A}$ would result in at least two vertices having the same distance vector. Therefore, $\mathbb{A}$ is a reduct. This completes the proof.
\end{prf}
The previous result yields that any reduct $\mathbb{A}$ and minimal resolving sets are equivalent which implies that reducts follow the same bounds as minimal resolving sets.
\par
It was proved in \cite{chart1} that $[\log_{3}(\Delta(\mathcal{G})+1)]\leq dim(\mathcal{G})$ and established in \cite{chart} that $dim(\mathcal{G}) \leq n-diam(\mathcal{G})$. A reduct must contain at least as many vertices as the metric dimension of the graph, since a smaller set would not be able to resolve all vertices. Thus, we have the following result which establishes a lower and upper bound on the size of a reduct in terms of structural properties of the connected graph \( \mathcal{G} \).
\begin{Proposition}
Let $\mathcal{G}$ be a non-trivial connected graph of order $n\geq 2$ and $\mathbb{A}$ be a reduct, then $[\log_{3}(\Delta(\mathcal{G})+1)]\leq |\mathbb{A}|\leq n-diam(\mathcal{G})$.

\end{Proposition}
\par
Reducts corresponding to minimal resolving sets with minimum cardinality are referred to as the metric basis, while reducts corresponding to minimal resolving sets of maximum cardinality are referred to as the upper basis. Minimum cardinality of minimal resolving sets is the metric dimension of $\mathcal{G}$, denoted as $dim(\mathcal{G})$ and maximum cardinality of minimal resolving set is the upper dimension of $\mathcal{G}$, denoted as $dim^{+}(\mathcal{G})$. It has been shown that for any pair of integers $a, b$, satisfying $ 2\leq a \leq b$, there exist a connected graph $\mathcal{G}$ with $dim(\mathcal{G})=a$ and $dim^{+}(\mathcal{G})=b$ \cite{gari}. Which establish that $dim^{+}(\mathcal{G})-dim(\mathcal{G})$ can be arbitrarily large, indicating that there exist reducts corresponding to such pair $(a, b)$ with an unbounded difference.
\par
We now establish bound on the size of a reduct \( \mathbb{A} \) of a connected graph \( \mathcal{G} \) in terms of its metric dimension and upper dimension. Since a reduct must contain at least as many vertices as the metric dimension of the graph, thus, we have $|\mathbb{A}| \geq \dim(\mathcal{G}).$ Since \( \mathbb{A} \) is also a resolving set (or a subset that induces a resolving set). The size of \( \mathbb{A} \) cannot exceed the upper dimension of the graph because the upper dimension \( \text{dim}^+(\mathcal{G}) \) is the maximum size of any resolving set. Therefore, we have the following proposition.
\begin{Proposition}
Let \( \mathcal{G} \) be a nontrivial connected graph of order \( n \geq 2 \), and let \( \mathbb{A} \) be a reduct. Then $
\dim(\mathcal{G}) \leq |\mathbb{A}| \leq \text{dim}^+(\mathcal{G})$.
\end{Proposition}
The following remark identifies the specific cases in which the given bound is attained exactly.
\begin{Remark}
Suppose $\mathbb{A}$ be a reduct, then for the complete graph \( K_n \), the sharp bound is $|\mathbb{A}|= \dim(K_n) = \text{dim}^+(K_n) = n - 1 $, for cycle \( C_n \), the sharp bound is $|\mathbb{A}|=\dim(C_n) = \text{dim}^+(C_n) = 2 \) for \( n \geq 3 $ and for complete bipartite graphs \( K_{m,n} \), the sharp bound is $|\mathbb{A}|=\dim(K_{m,n}) = \text{dim}^+(K_{m,n}) = \min\{m, n\}$.
\end{Remark}
\par
The following results characterize the reducts for families of graphs containing twin vertices.
\begin{Theorem}\label{thmr}
For $\mathcal{G}$, with distance-similar classes $\mathcal{B}_{1}, \mathcal{B}_{2}\cdots \mathcal{B}_{k}$ any reduct set $\mathbb{A}$ of $\mathcal{I}(\mathcal{G})$ contains at least $|\mathcal{B}_{i}|-1$ vertices from each $\mathcal{B}_{i}$ for all $i\in \{1, 2, \cdots k\}$.
\end{Theorem}
\begin{proof}
Let $\mathbb{A}$ be a reduct of $\mathcal{I}(\mathcal{G})$. Case i: Assume contrary $|\mathbb{A} \cap \mathcal{B}_{i}| < |\mathcal{B}_{i}|-1$ then there exist $v_i, v_j \in \mathcal{B}_{i}$ such that $v_i, v_j \not \in \mathbb{A}$. As $d(v_i, w) = d(v_j, w)$, for every vertex $w \in V \setminus \mathcal{B}_{i}$, therefore $\gamma(v_i|\mathbb{A}) = \gamma(v_j|\mathbb{A})$ This contradicts our assumption of
$\mathbb{A}$ being a reduct set. Case ii: Suppose for the contradiction $\mathbb{A}$ is not minimal then there exist $u \in \mathbb{A}$ such that $\mathbb{A}\setminus \{u\}$ is still a reduct. Consider $v_i, v_j \in V\setminus \mathbb{A}$, such that $v_i$ and $v_j$ are distance-similar then $\gamma(v_i|\mathbb{A}\setminus \{u\}) = \gamma(v_j|\mathbb{A}\setminus \{u\})$ which contradicts the definition of a reduct.
\end{proof}
By the previous theorem we can directly deduce the following results.
\begin{Corollary}
Let $\mathbb{A}$ be a reduct of $\mathcal{I}$, then $|\mathbb{A}|\leq |V|-k$ where $k$ is the count of distance-similar classes.
\end{Corollary}
\begin{Corollary}
If $v_{i}, v_{j} \in \mathcal{B}_i$ then at least one of $v_{i}$ and $v_{j}$ belongs to every reduct set.
\end{Corollary}
\begin{Remark}
Note that for $\mathcal{I}(\mathcal{G})$, where $\mathcal{G}=P_n$ is a path of $n$ vertices, then $RED= \{\{v_1\}, \{v_n\}, \{v_i, v_j\}\},$ where $i, j = \{2,\cdots, n-1\}$. For $\mathcal{G}=C_n$ is a cycle of $n$ vertices, $RED=\{v_i, v_j: v_i$ and $v_j$ are not antipodal vertices$\}$. For $\mathcal{G}=K_n$, is a complete graph of $n$ vertices, then $RED=V\setminus \{v_i\}$ for any $v_i\in V$.
\end{Remark}

The next result shows the connection between maximum partitioner and resolving sets (reducts) of $\mathcal{I}$ associated with $\mathcal{G}$.
\begin{Proposition}\label{max1}
For $\mathcal{I}$, let $\mathbb{A}\subseteq V$ we have;\\
(i) If $\mathbb{A}$ is a resolving set (i.e., $\mathbb{A} \in RED(\mathcal{I}))$ then $Max(\mathbb{A})=V$.\\
(ii) If $\mathbb{A}$ is not a resolving set (i.e., $\mathbb{A}\not\in RED(\mathcal{I}))$ then $Max(\mathbb{A})=\mathbb{A}$.
\end{Proposition}
\begin{prf}
(i) The proof follows from the Proposition \ref{max}.\\
(ii) Suppose $\mathbb{A}$ is not a resolving set then there does not exist $B\subseteq V$ such that $\pi_{B}=\pi_{\mathbb{A}}$ which implies that $Max(\mathbb{A})=\mathbb{A}$.
\qed
\end{prf}
In the political network considered earlier in Figure \ref{fig11}, if we consider the set $\mathbb{A}$ as containing agents such as John and Maria (John and Maria form a resolving set), their combined opinions can serve to resolve the differing viewpoints of all other agents in the network, therefore $Max(\mathbb{A})=\{John, Maria, Jessica, Emily,$ $David, Sarah\}$. Now if we consider the set $\mathbb{A}=\{John\}$ then $Max(\mathbb{A})=\{John\}$.
\par
The following result show that the automorphisms of the graph $\mathcal{G}$ preserve its reducts.
\begin{Proposition}
For \( \mathcal{G} \), let \( \mathbb{A} \subseteq V \) and \( \phi \) is an automorphism of \( \mathcal{G} \), then \( \mathbb{A} \in RED \) if and only if \( \phi(\mathbb{A}) \in RED\).
\end{Proposition}

\begin{proof}
The proof follows directly from Proposition \ref{auto}.
\end{proof}

\par
So far, we have studied general networks, now we consider a particular class of networks, where networks are modeled by commutative rings.  In this context, we use the algebraic structure of commutative rings to represent and analyze the networks. Vertices of the network can be linked to elements of a commutative ring, and edges are defined by an operation where two vertices are adjacent if their product is zero under modulo $n$. This framework aligns with the political network example (Figure \ref{fig11}), where the vertices represents individuals with binary opinion vectors, and edges were formed between vertices when their product (interpreted as elements in a commutative ring) was zero. This zero-product condition models opposition or disagreement between individuals, providing a concrete example of how commutative rings can be applied to represent networks with meaningful real-world interpretations.
\par
We now consider two specific families of zero-divisor graphs: one representing a family of graphs with twin vertices, and the other representing a family of graphs with twin-free vertices. Our goal is to investigate reducts for these families, explore their properties, and identify bounds for the reducts, which will help us understand the structural characteristics and distinctions between these graph families.
\par
We begin by examining the indiscernibility partitions within the vertex set of zero-divisor graphs corresponding to commutative rings $\mathbb{Z}_{n}$ and $\prod_{i=1}^{n}\mathbb{Z}_{2}$. In \cite{hibba}, it is studied the structure of the zero-divisor graphs of $\mathbb{Z}_{n}$, which represent a family of graphs with twin vertices. Let $D=\{d_{1}, d_{2}, \cdots d_{k}\}$ be the ordered set of divisors of $n$, since the number of distance-similar classes in $\Gamma(\mathbb{Z}_{n})$ depends on the number of divisors of $n$, we let $\mathcal{B}_1, \mathcal{B}_2, \cdots, \mathcal{B}_k$ be all distance-similar classes in $\Gamma(\mathbb{Z}_{n})$.
\par
The following proposition establishes the relationship between vertices, divisors of $n$, and distances in the graph $\Gamma(\mathbb{Z}_{n})$.
\begin{Proposition}\label{prop2}\cite{mitu}
For $\Gamma(\mathbb{Z}_{n})$ where $n=\prod_{j=1}^{m} p_{j}^{r_{j}}$ with $r_{j}\geq 1$, we have:\\
(i) For $v_{i}, v_{j}\in V$ if $v_{i}\cdot v_{j}=n$ then distance between $v_{i}$ and $v_{j}$ is 1.\\
(ii) The distance between any two vertices in $\mathcal{B}_i$ is 2.\\
(iii) For $v_{i}, v_{j}\in V$ if $d(v_{i}, v_{j})=a$ for $v_{i}\in \mathcal{B}_i$ and $v_{j}\in \mathcal{B}_j$ such that $i\neq j$ then $d(v_{i}, v_{k})=a$ for all $v_{k}\in \mathcal{B}_j$.
\end{Proposition}
Proposition \ref{prop2} establishes the foundational relationships between vertices, divisors, and distances within $\Gamma(\mathbb{Z}_{n})$. This insight is important for understanding the partitioning of the vertex set into distance-similar classes, which is a key aspect of our analysis of reducts for graphs with twin vertices.
\par
Note that, if $\mathbb{A}\subseteq\mathcal{B}_i$ and $B\subseteq\mathbb{A}^{c}$ then $v_{i}\equiv_{B} v_{j} \, \forall\, v_{i}, v_{j} \in \mathbb{A}$.  Furthermore, if $\mathbb{A}\subseteq\mathcal{B}_i$ and $v_{i}, v_{j} \in \mathcal{B}_i$ then $v_{i}\equiv_{\mathbb{A}} v_{j}$. Similarly, if $\mathbb{A}\cap \mathcal{B}_i\neq \emptyset$ then $\forall\, v_{i}, v_{j} \in \mathcal{B}_i \setminus \mathbb{A}$ we have $v_{i}\equiv_{\mathbb{A}} v_{j}$.

\par
The following theorem provides the count of elements within distance-similar classes.
 \begin{Theorem}\label{cardinality}\cite{winer}
For any divisor \( d_i \) of \( n \), the number of elements in the distance-similar class \( \mathcal{B}_{i} \) is given by \( \phi\left(\frac{n}{d_i}\right) \), where \( \phi \) denotes Euler’s totient function.
\end{Theorem}
\par
In the following results, we give the partition structure induced by subsets of distance-similar equivalence classes. For this we define $H_{1}, H_2, H_3$ as follows: $H_{1}=\{v_{i}: d(v_{i}, v_{j})=1 \,\forall v_j \in \mathbb{A}\}$, $H_{2}=\{v_{i}: d(v_{i}, v_{j})=2 \,\forall v_j \in \mathbb{A}\}$ and $H_{3}=\{v_{i}: d(v_{i}, v_{j})=3 \, \forall v_j \in \mathbb{A}\}$.

\begin{Proposition}
For $\Gamma(\mathbb{Z}_{n})$, where $n=\prod_{j=1}^{m} p_{j}^{r_{j}}$ with $r_{j}\geq 1$, let $\mathbb{A}=\{v_1, v_2, \cdots v_s\}\subseteq\mathcal{B}_i$ then $\pi_{\mathbb{A}}=v_{1}| \cdots|v_{s}| H_1|H_{2}|H_{3}$.
\end{Proposition}
\begin{prf}
Suppose $\mathbb{A}\subseteq\mathcal{B}_i$ then for any $v_i, v_j \in \mathbb{A}$ we have $v_{i}\not \equiv_{\mathbb{A}} v_{j}$. If $v_i\in \mathbb{A}$ and $v_j\notin\mathbb{A}$ then by above argument, representation of $v_i$ with respect to $\mathbb{A}$ is unique so $v_{i}\not \equiv_{\mathbb{A}} v_{j}$. Suppose $v_i, v_j \in H_{1}$ or $v_i, v_j \in H_{2}$ or $v_i, v_j \in H_3$ then by definition of $H_1, H_2$ and $H_3$, $d(v_i, v_k)=d(v_j, v_k)$ for all $v_k \in \mathbb{A}$ implies that $v_{i}\equiv_{\mathbb{A}} v_{j}$.
\qed
\end{prf}
\par
Note that for \( \Gamma(\mathbb{Z}_{n}) \), where \( n = p_{j}^{2} \), suppose \( \mathbb{A} = \{v_1, v_2, \cdots, v_s\} \subseteq V \) and \( |\mathbb{A}| < n - 1 \). If \( \gamma(v_i | \mathbb{A}) = (\cdots, 1, 1, \cdots) = \gamma(v_j | \mathbb{A}) \) for all \( v_i, v_j \in \mathbb{A}^{c} \), then \( v_i \equiv_{\mathbb{A}} v_j \). On the other hand, if \( v_i, v_j \in \mathbb{A} \), then \( \gamma(v_i | \mathbb{A}) = (\cdots, 0_{i^{th}}, \cdots) \) and \( \gamma(v_j | \mathbb{A}) = (\cdots, 0_{j^{th}}, \cdots) \), which implies that \( v_i \not\equiv_{\mathbb{A}} v_j \). Therefore, the distance-similar partition with respect to \( \mathbb{A} \) is given by $\pi_{\mathbb{A}} = v_1 | v_2 | \cdots | v_s | \mathbb{A}^{c}.$
\par
Let $T_{i}$ be the set of elements of $\prod_{i=1}^{n}\mathbb{Z}_{2}$ with $i$ non-zero
coordinates where $1 \leq i \leq n-1$, and let $t(v_{i})$ be the number of non-zero coordinates of $v_{i}$ in $\prod_{i=1}^{n}\mathbb{Z}_{2}$. Consider $R=\prod_{i=1}^{3}\mathbb{Z}_{2}$, we have $T_0 = \{(0,0,0)\}, T_1=\{ (1, 0, 0), (0, 1, 0), (0,0, 1)\}, T_2=\{(1, 1, 0), (1, 0, 1), (0, 1, 1)\}$ and $T_3 = \{(1,1,1)\}$. The number of non-zero coordinates of vertices are $t((0, 0, 0))=0$, $t((1, 0, 0))=1$, $t((0, 1, 0))=1$, $t((0, 0, 1))=1$, $t((1, 1, 0))=2$, $t((1, 0, 1))=2$ , $t((0, 1, 1))=2$ and $t((1, 1, 1))=3$.
\par
The following result demonstrates that the partitions of a vertex set induced by different subsets of vertices are equivalent.
\begin{Proposition}\label{t}
For $\Gamma(R)$, where $R=\prod_{i=1}^{n}\mathbb{Z}_{2}$ we have $\pi_{T_{i}}=\pi_{V}$ for all $1 \leq i \leq n-1$.
\end{Proposition}
\begin{proof}
Suppose $u, v\in T_{i}$ then there is zero in the representation of $u$ and $v$ on different positions (i.e, $d(u, u)=0\neq d(v, u)$) which implies that $u \not\equiv_{T_{i}}v$. Now suppose $u\in T_{i}$ and $ v\in V\setminus T_{i}$ then $t(u)<t(v)$ because there exist a zero in the representation of $u$ implies that $u\not\equiv_{T_{i}}v$. Lastly suppose $u, v\in V\setminus T_{i}$ such that $u\in T_{l}$ and $ v\in T_{m}$ then number of one's in the representation of $u$ is different from $v$, if $l<m$ then the number of one's in the representation of $u$ is greater than the number of one's in the representation of $v$ and vice vera. Hence $u\not\equiv_{T_{i}}v$. This completes the proof.
\end{proof}
\par
So far in this section, we have discussed indiscernibility partitions of networks associated with commutative rings. Now, we turn our attention to investigating reducts (minimal resolving sets) within networks associated with commutative rings. Specifically, we relate the concept of reducts to the metric dimension and the upper dimension of these networks.
\par
In \cite{pirzada2}, it is established that for a finite commutative ring \( R \) of order \( 2k \), where \( k \) is an odd integer, the metric dimension \( \dim(\Gamma(R)) \) and the upper dimension \( \dim^+(\Gamma(R)) \) of the associated zero-divisor graph \( \Gamma(R) \) are identical. Furthermore, in the case of a reduct, the size of the reduct \( |\mathbb{A}| \) is equal to both the metric dimension \( \dim(\Gamma(R)) \) and the upper dimension \( \dim^+(\Gamma(R)) \) of \( \Gamma(R) \). The same paper also confirms that for a zero-divisor graph associated with a finite commutative ring \( R \) of order \( 2k \), where \( k \) is odd, the metric and upper dimensions remain equal.
In \cite{pirzada3} Pirzada established that for zero divisor graphs containing twin vertices, the upper dimension is equal to the metric dimension. In \cite{redmond}, it was established that for \( n \geq 2 \) and \( R = \prod_{i=1}^{n} \mathbb{Z}_2 \), the metric dimension and upper dimension of \( R \) are equal if and only if \( n = 4 \).
\par
The characteristic of a commutative ring \( R \) is the smallest positive integer \( k \) such that \( kr = 0 \) for all \( r \in R \). If no such \( k \) exists, the ring is said to have infinite characteristic. As shown in \cite{pirzada2}, for a finite commutative ring \( R \) that is not a field and has an odd characteristic, the upper dimension \( \dim^+(\Gamma(R)) \) of the associated zero-divisor graph \( \Gamma(R) \) is equal to its metric dimension \( \dim(\Gamma(R)) \). Furthermore, in this case, the size of the reduct \( |\mathbb{A}| \) is also equal to both \( \dim^+(\Gamma(R)) \) and \( \dim(\Gamma(R)) \).
\par
In \cite{raja}, Raja and Pirzada discuss the properties of minimal locating sets in the context of zero-divisor graphs. In particular, they examine the behavior of these sets when the graph \( \Gamma(R) \) has a cut vertex. Let \( R \) be a finite commutative ring with unity, such that the vertex set \( |V| \geq 3 \). They establish that if \( \Gamma(R) \) contains a cut vertex, say \( w \), then the minimal locating set or reduct \( \mathbb{A} \) can be characterized in one of two ways:
\begin{itemize}
    \item[(i)] \( \gamma(w|\mathbb{A}) = (1, 1, \dots, 1, d(x, v_1), d(x, v_2), \dots, d(w, v_{l-m+1})) \), where \( l = |\mathbb{A}| \) and \( m \) is the number of degree one vertices incident on \( w \).
    \item[(ii)] Alternatively, the size of the reduct is 3, i.e., \( |\mathbb{A}| = dim(\Gamma(R)) = 3 \).
\end{itemize}
The following remark provides an upper bound on the size of a reduct $\mathbb{A}$ for the graph $\Gamma(R)$ where $R=\mathbb{Z}_{n}$.
\begin{Remark}
Note that for $\Gamma(R)$ where $R=\mathbb{Z}_{n}$, and $\mathbb{A}$ be a reduct of $\mathcal{I}$, then $|\mathbb{A}|\leq n-\phi(n)-1-|D|$, where $D$ is a set of non-trivial divisors of $n$ and $\phi$ denotes Euler’s totient function.
\end{Remark}
The following theorem gives the reduct for the zero-divisor graph $\Gamma(\prod_{i=1}^{n}\mathbb{Z}_{2})$ with $n\geq 5$.
\begin{Theorem}\cite{redmond}
For $\Gamma(\prod_{i=1}^{n}\mathbb{Z}_{2})$ with $n\geq 5$, let $T_1, T_{n-1}$ and $T_{k}\setminus \{w\}$ are resolving sets (reduct), where $w\in T_k$, $1< k< n-1$ and $k\neq \frac{n}{2}$.
\end{Theorem}
\begin{Corollary}
For $\Gamma(R)$, where $R=\prod_{i=1}^{n}\mathbb{Z}_{2}$, let $\mathbb{A}$ be a reduct then for $n\leq 3, |\mathbb{A}\cap T_{i}|=|T_i|-1\text{ for all }1 \leq i \leq n-1\text{ and }\mathbb{A}\cap T_{j}=\emptyset \text{ for all } \,i\neq j$, for $n\leq 3, |\mathbb{A}\cap T_{1}|=1 \text{ and } |\mathbb{A}\cap T_{2}|=1 \text{ for any } v_i, v_j \in \mathbb{A}, v_i\nsim v_j$, for $n=4,|\mathbb{A}\cap T_{n-1}|=|T_{n-1}|-1$ and for $n\geq4, \mathbb{A}=T_{1}$.
\end{Corollary}
Algorithm 1 is based on the definition of reduct and used to find the reducts (minimal resolving sets) of $\mathcal{G}$.
\floatname{algorithm}{Algorithm}
\begin{algorithm}[h]\label{alg3}
\caption{Reducts of an information system $\mathcal{I}$ associated with $\mathcal{G}$}
\begin{algorithmic}
\State \algorithmicrequire $R=\{X_{i}\subseteq V: 1\leq i\leq s\wedge s=2^{|V|}\}$\algorithmiccomment{
Collection of all subsets of $V$.}
\State \algorithmicensure $R$ \algorithmiccomment{Collection of all reducts of $\mathcal{I}$.}%
\State Initialize: $i = 1$
\For{$i \leq s$}
\If{$\pi_{X_{i}}\neq\pi_{V}$}
        \State $R=R\setminus X_{i}$
        \State $i=i+1$
\Else
        \State $i=i+1$
\EndIf
\EndFor
\State $R=\{A_{j}: 1\leq j\leq k\}$ \algorithmiccomment{$R$ updated from the above for loop}
\State $j=1$
\For{$j \leq k$}
\State $H_{j}=\{Y_{l}\subset A_{j} : 1\leq l\leq r\}$\algorithmiccomment{Collection of all proper subsets of $A_{j}\in R$}
        \If{$\pi_{Y_{l}}=\pi_{A}$, for at least one $l\in\{1,2,\cdots,r\}$ }
        \State $R=R\setminus A_{j}$
        \State $j=j+1$
\Else
        \State $j=j+1$
\EndIf
\EndFor
\State Return $R$
\end{algorithmic}
\end{algorithm}
\subsection{Dependency Measures and Approximations}
The partial order relation between two set partitions represents the hierarchy of their respective indiscernibility blocks and is closely associated with the concepts of the positive region and dependency measure. For \( \mathbb{A}, \mathbb{B} \subseteq V \), the positive region of \( \mathbb{A} \) with respect to \( \mathbb{B} \) is given by $POS_{\mathbb{B}}(\mathbb{A}) = \{ v_j \in V \mid [v_j]_{\mathbb{B}} \subseteq [v_j]_{\mathbb{A}} \}.$
The degree of dependency of $\mathbb{A}$ on $\mathbb{B}$ is defined as $\kappa_{\mathbb{B}}(\mathbb{A})= |POS_{\mathbb{B}}(\mathbb{A})|/|V|$, which ranges from 0 to 1. In particular, if $\mathbb{A}, \mathbb{B}\subseteq V$ and $\pi_{\mathbb{A}}\preceq \pi_{\mathbb{B}}$ then $POS_{\mathbb{A}}(\mathbb{B})= V$. This is because, when the partition induced by \( \mathbb{A} \) is finer than the partition induced by \( \mathbb{B} \), every element \( x \) can be fully identified (or positively classified) by the information in \( \mathbb{B} \), since \( [x]_{\mathbb{A}} \subseteq [x]_{\mathbb{B}} \). Consequently, the dependency measure \( \gamma_{\mathbb{A}}(\mathbb{B}) \) is equal to 1, as the positive region contains all the elements of \( V \).

\par
The next proposition analyze the positive region and dependency relationships for different vertex subsets within the information system \( \mathcal{I} \).
\begin{Proposition}\label{pos}
For $\mathcal{I}$ if $\mathbb{A}$ and $\mathbb{B}$ are resolving sets then $POS_{\mathbb{A}}(\mathbb{B})=V$.
 \end{Proposition}
 \begin{prf}
Suppose $\mathbb{A}$ and $\mathbb{B}$ are resolving sets then $\pi_{\mathbb{A}}=v_1|v_2| \cdots| v_n=\pi_{\mathbb{B}}$, for all $v_{i}\in V$,  $C_{\mathbb{A}}(v_{i})\subseteq C_{\mathbb{B}}(v_{i})$ implies that $POS_{\mathbb{A}}(\mathbb{B})=V$.
\qed
\end{prf}
The next result gives the positive region for two distinct subsets within a distance-similar class.
\begin{Proposition}
For non-empty subsets $\mathbb{A}, \mathbb{B}$ of $V$ we have\\
(i) For $\mathbb{A}, \mathbb{B}\subseteq \mathcal{B}_i$, if $\mathbb{A}\nsubseteq\mathbb{B}$ and $\mathbb{A}\nsupseteq\mathbb{B}$ then $POS_{\mathbb{A}}(\mathbb{B})=\emptyset$ and $POS_{\mathbb{B}}(\mathbb{A})=\emptyset$.\\
(ii) For $\mathbb{A}=\mathcal{B}_i$ and $\mathbb{A}\supset\mathbb{B}$ then $POS_{\mathbb{A}}(\mathbb{B})=V$.
\end{Proposition}
\begin{prf}
(i) Suppose $\mathbb{A}, \mathbb{B}$  are non-empty subsets of $\mathcal{B}_i$ such that $\mathbb{A}\nsubseteq\mathbb{B}$ and $\mathbb{A}\nsupseteq\mathbb{B}$. Then there exist $v_i, v_j \in \mathbb{A}$ and $v_i, v_j \not \in \mathbb{B}$ such that $v_i\not \equiv_{\mathbb{A}} v_j$ and $v_i\equiv_{\mathbb{B}} v_j$. Therefore $POS_{\mathbb{A}}(\mathbb{B})=\emptyset$ and $POS_{\mathbb{B}}(\mathbb{A})=\emptyset$.\\
(ii) Suppose $\mathbb{A}=\mathcal{B}_i$ and $\mathbb{A}\supset\mathbb{B}$ then by definition of distance-similarity $\pi_{\mathbb{A}}\preceq \pi_{\mathbb{B}}$ implies that $POS_{\mathbb{A}}(\mathbb{B})= V$.
\qed
\end{prf}
\par
Note that, if $\mathbb{A} = \emptyset$ and $\mathbb{B}= \emptyset$ then $\pi_{\mathbb{A}}=\pi_{\mathbb{B}}=V$, therefore $C_{\mathbb{A}}(v_{i})=V= C_{\mathbb{B}}(v_{i})=V$ for all $v_{i}\in V$ which implies that $POS_{\mathbb{A}}(B)=V$ and $\gamma_{\mathbb{A}}(B)=1$. Similarly, if $\mathbb{A} = \emptyset$ and $\mathbb{B}\neq \emptyset$ then $\pi_{\mathbb{A}}=V$ and $\pi_{\mathbb{B}}= C_{\mathbb{B}}(v_{1})|C_{\mathbb{B}}(v_{2})| \cdots| C_{\mathbb{B}}(v_{l})$. Thus $\pi_{\mathbb{B}}$ is finer than $\pi_{\mathbb{A}}$ implies that $\pi_{\mathbb{B}}\subseteq \pi_{\mathbb{A}}$ this implies $POS_{\mathbb{A}}(\mathbb{B})=\emptyset$ and $\gamma_{\mathbb{A}}(B)=0$.
Now let $\mathbb{B} \subseteq \mathbb{A}$ then by above $\pi_{\mathbb{A}}\subseteq \pi_{\mathbb{B}}$ and $\gamma_{\mathbb{A}}(B)=1$.
\par
The next result examines the positive region of two distinct subsets within the vertex set of $ \Gamma( \prod_{i=1}^{k} \mathbb{Z}_2)$.
\begin{Proposition}
For $ \Gamma( \prod_{i=1}^{k} \mathbb{Z}_2)$, we have $POS_{T_{i}}(T_j) = V $.
\end{Proposition}
\begin{prf}
Let \( \mathbb{A} = T_i \) for some \( i \in \{1, 2, \dots, k-1\} \). Then, from Proposition \ref{t}, \( \pi_{\mathbb{A}} = \pi_V \). Therefore, \( \pi_{T_i} = \pi_{T_j} \), which implies $POS_{\mathbb{A}}(T_j) = V.$
\end{prf}
\par

RST provides a framework for dealing with incomplete or uncertain information within a network. RST can help approximate the network structure by identifying subsets of nodes or edges that are relevant to certain properties or behaviors of interest. The lower approximation can be applied to define the core nodes that belong to a certain cluster, community, or network substructure with certainty. For instance, if we have a community detection problem in a social network, the lower approximation would include nodes that are strongly connected or exhibit high similarity with others in the same community. The upper approximation would include nodes that are on the boundary or have weaker ties, such as those that connect across multiple communities or are loosely associated with the core nodes of a cluster.
\par
We define lower and upper approximations for the vertex subset $X$.
\begin{defn}
For any $\mathbb{A}\subseteq V$, we call $\mathcal{L}_{\mathbb{A}}(X)= \{v_{i} \in V: C_{\mathbb{A}}(v_{i}) \subseteq X\}$ and $\mathcal{U}_{\mathbb{A}}(X)= \{v_{i} \in V: C_{\mathbb{A}}(v_{i})\cap X\neq\emptyset\}$ the lower and upper approximation of $X\subseteq V$. Any $X\subseteq V$ is called $\mathbb{A}$-definable/$\mathbb{A}$-exact if $\mathcal{L}_{\mathbb{A}}(X)=\mathcal{U}_{\mathbb{A}}(X)$, otherwise the set is called $\mathbb{A}$-rough.
\end{defn}
\par
The approximations discussed above exhibit several fundamental properties. Specifically, for any subsets \( X, Y \subseteq V \), the relation \( \mathcal{L}_{\mathbb{A}}(X) \subseteq X \subseteq \mathcal{U}_{\mathbb{A}}(X) \) holds. Moreover, the lower approximation of the complement of \( X \) is equal to the complement of the upper approximation of \( X \), and vice versa. Additionally, if \( X \) is a subset of \( Y \), then the lower approximation of \( X \) is contained within the lower approximation of \( Y \), and the same property holds for upper approximations.

\par
In the next result, we interpret how two distinct resolving sets can yield the same lower and upper approximations.
\begin{Proposition}
For $\mathcal{I}$, suppose $\mathbb{A}_{1}$ and $\mathbb{A}_{2}$ are two resolving sets then for any $X \subseteq V$, $\mathcal{L}_{\mathbb{A}_{1}}(X)= \mathcal{L}_{\mathbb{A}_{2}}(X)$ and $\mathcal{U}_{\mathbb{A}_{1}}(X)=\mathcal{U}_{\mathbb{A}_{2}}(X)$.
\end{Proposition}
\begin{prf}
Suppose $\mathbb{A}_{1}$ and $\mathbb{A}_{2}$ are two resolving sets, let $X\subseteq  V$ and for all $v_{i}\in V$, if $C_{\mathbb{A}_{1}}(v_{i}) \subseteq X$ then by the resolving set property $\pi_{\mathbb{A}_{1}}=v_1|v_2|\cdots|v_n =\pi_{ \mathbb{A}_{2}}$ which implies that $C_{\mathbb{A}_{2}}(v_{i}) \subseteq X$. Hence $\mathcal{L}_{\mathbb{A}_{1}}(X)= \mathcal{L}_{\mathbb{A}_{2}}(X)$. Now for all $v_{i}\in V$, if $C_{\mathbb{A}_{1}}(v_{i}) \cap X\neq\emptyset$ then by definition of reduct $\pi_{\mathbb{A}_{1}}=\pi_{ \mathbb{A}_{2}}$ implies that $C_{\mathbb{A}_{2}}(v_{i}) \cap X\neq\emptyset$. Hence $\mathcal{U}_{\mathbb{A}_{1}}(X)=\mathcal{U}_{\mathbb{A}_{2}}(X)$.
\qed
\end{prf}
\par
In the following result, we establish the relationship between \( \mathbb{A} \) and \( X \) when \( X \) is \( \mathbb{A} \)-exact.
\begin{Proposition}
For $\mathcal{I}$, suppose $\mathbb{A}, X\subseteq V$ then we have:\\
(i) If $X\subseteq \mathbb{A}$ then $X$ is $\mathbb{A}$-exact.\\
(ii) $X\subseteq V$ is $\mathbb{A}$-exact if $\pi_{\mathbb{A}}=\pi_{V}$.
\end{Proposition}
\begin{prf}
(i) Let $X\subseteq \mathbb{A}$ then for all $v_i \in \mathbb{A}$, $C_{\mathbb{A}}(v_{i})=\{v_i\}$ then $\mathcal{L}_{\mathbb{A}}(X)=X=\mathcal{U}_{\mathbb{A}}(X)$.\\
(ii) Let $X\subseteq V$ and $\pi_{\mathbb{A}}=\pi_{V}$ then all the classes are singletons implies $\mathcal{L}_{\mathbb{A}}(X)=X=\mathcal{U}_{\mathbb{A}}(X)$ hence $X$ is $\mathbb{A}$-exact.
\qed
\end{prf}
In the previous section, we explored partitions and introduced a method based on the definition of reducts to identify all resolving sets. Now, we propose an alternative approach that utilizes the discernibility function derived from the distance-based discernibility matrix to determine all minimal resolving sets.
\section{The Discernibility Matrix and Resolving sets}
In this section, we propose a method of finding all minimal resolving sets of a graph using the notion of discernibility matrix. For this purpose, we first introduce the distance-based discernibility relation on the vertex set $V$, two vertices $v_{i}$ and $v_j$ in $V$ are said to be discernible with respect to a vertex $w$ denoted as $v_{i} \not\equiv_{w} v_j$ if and only if $d(v_i, w)\neq d(v_j, w)$. Based on this, we define the distance-based discernibility matrix $\Delta_{V}$ associated with the distance-based discernibility relation. In $\Delta_{V}$, each entry corresponds to the vertices that resolve a given pair of vertices based on their distance.
\begin{defn}\label{def}
For $\mathcal{I}$, the distance-based discernibility matrix $\Delta_{V}$ is the $ |V| \times |V|$ matrix, and the entries are $\Delta_{V}(v_{i}, v_j)$ where
\begin{equation}\label{2}
\Delta_{V}(v_{i}, v_j)=\{w\in V: d(v_{i}, w) \neq d(v_j, w) \}.
\end{equation}
\end{defn}
Note that $\Delta_{V}$ is symmetric, meaning $\Delta_{V}(v_{i}, v_j) = \Delta_{V}(v_j, v_{i})$ and $\Delta_{V}(v_{i}, v_{i}) = \emptyset$. For the remainder of this section, we will consider a lower triangular matrix. Vertices with same entries in the discernibility matrix may be structurally equivalent, meaning they play similar roles in the network. \( DISC(\mathcal{I}) \) represent the set of all distinct entries in the distance-based discernibility matrix \( \Delta_{V} \).
\par
The relationship between the entries in the distance-based discernibility matrix and the indiscernibility of vertices can be described through the following properties.
Suppose $\mathbb{A}\subseteq V$ and $v_{i}, v_j\in V$, we have\\
(i) If $\mathbb{A}= \Delta_{V}(v_{i}, v_j)$ then $v_i\equiv_{V\setminus \mathbb{A}} v_j$.\\
(ii) If $v_{i}\equiv_{V\setminus \mathbb{A}} v_j$ then $\Delta_{V}(v_{i}, v_j)\subseteq \mathbb{A}.$\\
(iii) $\Delta_{V}(v_{i}, v_j)\cap \mathbb{A}= \emptyset$ if and only if  $v_{i}\equiv_{\mathbb{A}} v_j$.\\
These properties reveal the fundamental relationship between discernibility and indiscernibility of vertices in a graph, as captured by the distance-based discernibility matrix. Together, they provide a structured way to understand how subsets of vertices contribute to distinguishing or grouping vertices based on their distance representations.
\par
It is evident from the Definition \ref{def} that the discernibility of two vertices depends on the distances between them. The following result provides a representation of the entries in the distance-based discernibility matrix of \( \mathcal{I} \) corresponding to the simple undirected graph \( \mathcal{G} \).
\begin{Theorem}
For $\mathcal{I}$, the entries of $\Delta_V$ are
\[
  \Delta_{V}(v_{i}, v_j)= \left.
  \begin{cases}
    V \setminus \{X_{t}(v_i)\cap X_{t}(v_j)\}, & \text{if }  v_i\neq v_j  \\
    \emptyset , & \text{if }  v_i=v_j
    \end{cases}
  \right.
\]
where $X_{t}(v_i)=\{w\in V: d(v_i, w)=t\}$ for all $t\in \{1, 2, 3\cdots diam(\mathcal{G})\}.$
\end{Theorem}
\begin{proof}
(i) Let $ v_i\neq v_j$, suppose $k\in X_{t}(v_i)$ and $k\in X_{t}(v_j)$ where $t$ is fixed then $d(v_i, k)=t=d(v_j, k)$ by definition of distance-based discernibility matrix $k\notin \Delta_{V}(v_{i}, v_j)$. Now suppose $k\in X_{t}(v_i)$ and $k\in X_{s}(v_j)$ then $d(v_i, k)=t\neq d(v_j, k)=s$ implies that $k\in \Delta_{V}(v_{i}, v_j)$ and we conclude that $\Delta_{V}(v_{i}, v_j)= V \setminus \{X_{t}(v_i)\cap X_{t}(v_j)\}$.\\
(ii) Let $ v_i= v_j$ (comparing vertex with itself), then $d(v_i, k)=d(v_j, k)$ for $k\in V$ implies that $\Delta_{V}(v_{i}, v_j)=\emptyset$.
\end{proof}
For $\mathcal{G}$ with distance-similar vertices, we establish the following result.
\begin{Corollary}\label{cor}
For $\mathcal{G}$, if all the vertices are distance-similar then $\Delta_{V}(v_{i}, v_j)=\{v_i, v_j\}$ for all $v_i, v_j\in V$.
\end{Corollary}
\begin{Remark}
For the information table $\mathcal{I}$ associated with a path graph $P_n$ or a cycle graph $C_n$, for any pair of vertices $v_i$ and $v_j$ in the vertex set $V$, the entries of $\Delta_V(v_i, v_j)$ is defined as the set of vertices $V \setminus \{w\}$ such that $d(v_i, w)=d(v_j, w)$ for any $w\in V$. In the case of complete graph, for any $v_i$ and $v_j$ in $K_n$, the entries of $\Delta_V(v_i, v_j)=\{v_i, v_j\}$, since in a complete graph, all vertices are equidistant from one another, with the distance being 1.
\end{Remark}
\par
An information table can be associated with a numerical discernibility matrix, where each entry indicates the number of vertices that can distinguish the corresponding pair of vertices.
\begin{defn}
The numerical discernibility matrix associated with $\mathcal{I}$ is denoted by $\mathcal{N}_{V}$ and the entries in $ \mathcal{N}_{V}$ are defined as $\mathcal{N}(v_{i}, v_j) =|\Delta_{V}(v_{i}, v_j)|$.
\end{defn}
\begin{Remark}
If \( \mathcal{N}(v_{i}, v_j) > 0 \), there exists at least one vertex in \( V \) that distinguishes \( v_i \) from \( v_j \) based on their distances. If \( \mathcal{N}(v_{i}, v_j) = 0 \), the pair \( v_i \) and \( v_j \) is indistinguishable, as their distances to all other vertices in \( V \) are identical.
\end{Remark}
\par The following result provides the bounds for the entries in the numerical discernibility matrix associated with $\mathcal{I}$.
\begin{Proposition}
For $\mathcal{I}$, the minimum non-zero entry in $\mathcal{N}_{V}$ is 2 and maximum entry in $\mathcal{N}_{V}$ is $|V|$.
\end{Proposition}
\begin{prf}
Suppose $v_i, v_j \in V$ such that $v_i$ and $v_j$ are distance similar then $\Delta_{V}(v_{i}, v_j)=\{v_i, v_j\}$ implies that minimum entry in $\mathcal{N}_{V}$ is 2. If there exist $v_i, v_j \in V$ such that $v_i$ and $v_j$ are on a diametral path then $\Delta_{V}(v_{i}, v_j)=V$ and maximum entry in $\mathcal{N}_{V}$ is $|V|$.
\qed
\end{prf}
\par
\par
We now propose a novel approach for finding all minimal resolving sets in graphs by using the techniques of attribute reduction which involves discernibility function and distance-based discernibility matrix. The discernibility function helps in finding these sets by using the elements of the distance-based discernibility matrix of an information system.
\begin{defn}
The discernibility function $\zeta_{dis}$ for the information table $\mathcal{I}$ is defined as follows:
\begin{center}
$\zeta_{dis}=\wedge\{\vee \Delta_{V}(v_{i}, v_j): \Delta_{V}(v_{i}, v_j)\neq \emptyset\}$,
\end{center}
where $\vee \Delta_{V}(v_{i}, v_j)$ denote the disjunction of all vertices in $\Delta_{V}(v_{i}, v_j)$, which demonstrates that each vertex in $\Delta_{V}(v_{i}, v_j)$ can distinguish $v_{i}$ and $v_j$. The formula $\wedge\{\vee \Delta_{V}(v_{i}, v_j)\}$ is the conjunction of all $\vee \Delta_{V}(v_{i}, v_j)$.
\end{defn}
The discernibility function generates attribute subsets (minimal resolving sets) that distinguish all elements of $\mathcal{I}$. To find minimal elements in the distance-based discernibility matrix without computing the entire matrix, we can focus on an incremental approach. This involves examining only necessary pairs of objects to identify which attributes are essential for discernibility. This can be achieved by iterating over pairs of objects and updating a list of necessary attributes dynamically. The Algorithm 2 computes all the resolving sets of the information table $\mathcal{I}$ using distance-based discernibility matrix.
\floatname{algorithm}{Algorithm}
\begin{algorithm}[h]\label{alg2}
\caption{Resolving sets of a graph using distance-based discernibility matrix}
\begin{algorithmic}
\State \algorithmicrequire $DISC(\mathcal{I})=\{\mathbb{A}_{i}: 1\leq i\leq s\}$ \algorithmiccomment{Collection of all the distinct entries of the distance-based discernibility matrix of $\mathcal{I}$.}
\State \algorithmicensure $RED$ \algorithmiccomment{Resolving set (reduct) of $\mathcal{I}$.}%
\State Initialize: $i = 1$
\For{$i \leq s$}
\State Initialize: $j = 2$
\For{$j \leq s$}
\If{$\mathbb{A}_{j}\subset \mathbb{A}_{i}$}
        \State $DISC(\mathcal{I})=DISC(\mathcal{I})\setminus \mathbb{A}_{i}$
        \State $j=j+1$
\Else
        \State $DISC(\mathcal{I})=DISC(\mathcal{I})$
        \State $j=j+1$
\EndIf
\EndFor
        \State $i=i+1$
\EndFor
\State Return $DISC(\mathcal{I})$
\State $DISC(\mathcal{I})=\{\mathbb{A}_{l}: 1\leq l\leq h\}$ \algorithmiccomment{updated $DISC(\mathcal{I})$ }
\State Initialize: $l=1$
\State Initialize: $RED=\emptyset$
\State Initialize: $R=\emptyset$
\For{$l \leq h$}
        \If{$R\cap \mathbb{A}_{l}=\emptyset$ }
        \State $R=R\cup\{v\}$, for some $v\in \mathbb{A}_{l}$
        \State $l=l+1$
\Else
        \State $RED=RED\cup R$
        \State $l=l+1$
\EndIf
\EndFor
\State Return $R$
\end{algorithmic}
\end{algorithm}

\par
The following example illustrates the process of identifying resolving sets within a graph using the discernibility function. This approach provides a structured method for finding all the minimal resolving sets in a graph by analyzing pairwise relationships between vertices and their distinguishing characteristics. The resolving sets in the following example coincide with those derived from previous sections.
\begin{Example}
Consider a common example in social network analysis the collaboration network among coworkers in a workplace. This network demonstrates how employees form connections through projects, shared tasks, or informal interactions. The nodes of the network are employees with two nodes being connected by an edge if the corresponding employees have previously collaborated on a project.
The vertex set is $V=\{Sarah, John, Emily, Michael, Ania\}$ and the edges are $(Sarah, John)$, $(John, Emily)$, $(Emily, Michael)$, $(Michael, Ania)$. The distinct nonzero entries in the distance-based discernibility matrix are given as follows:
\begin{center}
$DISC(\mathcal{I})=\{V, \{Sarah, Emily, Michael, Ania\}, \{Sarah, John, Michael, Ania\}, \{Sarah, John, Emily, Ania\}\}$.
\end{center}
We can find these resolving sets by using Algorithm 2. We start from applying discernibility function on $DISC(\mathcal{I})$ as
\begin{center}
$\zeta_{dis}= \{\{Sarah\vee John \vee Emily\vee Michael\vee Ania\}\wedge\{Sarah\vee Emily\vee Michael\vee Ania\}\wedge\{Sarah\vee John\vee Michael\vee Ania\}\wedge\{Sarah\vee John\vee Emily\vee Ania\}\}$.
\end{center}
By applying conjunction and disjunction, the obtained resolving sets are
\begin{center}

$Resolving\,\,Sets=\{\{Sarah\}, \{Ania\}, \{John, Emily\}, \{John, Michael\}, \{Emily, Michael\}\}$.
\end{center}
\end{Example}
\par
The following result, indicates that automorphisms preserve the structure of the discernibility matrix.
\begin{Proposition}
Let $\phi \in \text{Aut}(\mathcal{G})$ be an automorphism of a graph $\mathcal{G}$. Then the discernibility matrix $\Delta_{V}$ satisfies $\Delta_{\phi(V)}(\phi(v_i), \phi(v_j)) = \Delta_{V}(v_i, v_j),$ $\forall v_i, v_j \in V.$
\end{Proposition}

\begin{prf}
Let $\phi \in \text{Aut}(\mathcal{G})$, by definition of $\phi$, $d(v_i, v_j) = d(\phi(v_i), \phi(v_j)),$ $\forall v_i, v_j \in V.$ Since $\phi$ is a bijection, this directly implies
$\Delta_{\phi(V)}(\phi(v_i), \phi(v_j)) = \{ \phi(v_k) \mid v_k \in V, \, d(\phi(v_i), \phi(v_k)) \neq d(\phi(v_j), \phi(v_k)) \}=\Delta_V(v_i, v_j).
$
\qed
\end{prf}
\par
The next result establishes a connection between the resolving set and the entries of the distance-based discernibility matrix.
\begin{Remark}
For $\mathcal{I}$, let $\mathbb{A}\subseteq V$ be a resolving set and $v_{i}, v_j\in V$, we have\\
(i) $\Delta_{V}(v_{i}, v_j) \cap \mathbb{A} \neq \emptyset$, for all $v_i, v_j \in V$.\\
(ii) If $\mathbb{A}\subseteq \Delta_{V}(v_{i}, v_j)$, then $v_{i}$ and $v_j$ are resolved by $\mathbb{A}$.
\end{Remark}
\par
Next, we provide the representation of the entries in the distance-based discernibility matrix associated with \( \Gamma(\mathbb{Z}_{n}) \).
\begin{Theorem}
For $\Gamma(\mathbb{Z}_{n})$, the entries of $\Delta_V$ are
\[
  \Delta_{V}(v_{i}, v_j)= \left.
  \begin{cases}
   \{v_i, v_j\}, & \text{if }  v_i, v_j \in \mathcal{B}_i \\
\emptyset , & \text{if }  v_i=v_j\\
 (V\setminus \{w:gcd(w, v_i)\neq 1\wedge gcd(w. v_j)\neq 1\})\cup\{u:u\cdot v_i=0\vee u\cdot v_j=0\}\cup\{v_i, v_j\} & \text{otherwise }
\end{cases}
  \right.
\]
\end{Theorem}
\begin{proof}
(i) Let $v_i, v_j \in \mathcal{B}_i$ (i.e, $v_i$ and $v_j$ are distance-similar) then by Corollary \ref{cor} we have $ \Delta_{V}(v_{i}, v_j)=\{v_i, v_j\}$. \\
(ii) Let $ v_i= v_j$, then $d(v_i, k)=d(v_j, k)$ for $k\in V$ implies that $\Delta_{V}(v_{i}, v_j)=\emptyset$.\\
(iii) Let $v_i, v_j \in V$ and both are not distance-similar then we have three cases: Case 1. if there exist $w\in V$ such that $gcd(w, v_i)\neq 1$ and $gcd(w, v_j)\neq 1$ then $d(v_i, w)= d(v_j, w)$ so, $ \Delta_{V}(v_{i}, v_j)\subset V\setminus w$. Case 2. if  there exist $u\in V$ such that $gcd(u, v_i)\neq 1$ and $gcd(u, v_j)\neq 1$, $u\cdot v_i=0$ or $ u\cdot v_j=0$ then we add $u$ in $ \Delta_{V}(v_{i}, v_j)$. Case 3. Let $v_i, v_j \in V$ then $d(v_i, v_i)=0$ and $d(v_j, v_j)=0$ implies that $\{v_i, v_j\} \subseteq \Delta_{V}(v_{i}, v_j)$.
\end{proof}
\par
Next, we present the representation of the entries in the distance-based discernibility matrix corresponding to \( \Gamma\left(\prod_{i=1}^{n} \mathbb{Z}_{2} \right) \). Let us denote $\bar{u}$ as the complement of the vertex $u$. For example, in $\Gamma(\prod_{i=1}^{3}\mathbb{Z}_{2})$ if $u=(1, 0, 0)$ then $\bar{u}=(0, 1, 1)$. Moreover, if $u=(1, 1, 0)$ and $v=(0, 1, 1)$ then we have  $u+v=(1, 0, 1)$.
\begin{Theorem}
For $\Gamma(\prod_{i=1}^{n}\mathbb{Z}_{2})$, the entries of $\Delta_V$ are
\[
  \Delta_{V}(v_{i}, v_j)= \left.
  \begin{cases}
   N(v_i)\Delta N(v_j), & \text{if }  v_i, v_j \in T_1 \\
 N(v_i)\Delta N(v_j)\cup N(\bar{v_i})\Delta N(\bar{v_j})\cup \{v_i+v_j, \bar{v_i}+\bar{v_j}\} & \text{if }  v_i, v_j \in T_j\, \forall j\in\{2, 3,\cdots T_{n-2}\} \\
 N(\bar{v_i})\Delta N(\bar{v_j})& \text{if }  v_i, v_j \in T_n-1 \\
    \emptyset , & \text{if }  v_i=v_j
    \end{cases}
  \right.
\]
where $N(v_i)$ is the neighborhood of $v_i$ and $\Delta$ denotes the symmetric difference between sets.
\end{Theorem}
\begin{proof}
(i) Let $v_i, v_j\in T_1$, there exist $w\in V$ such that $w\in N(v_i)$ and $w\in N(v_j)$ then $d(v_i, w)=1=d(v_j, w)$ implies that $w\not\in \Delta_{V}(v_{i}, v_j)$. Now if $w\in N(v_i)$ but $w\not\in N(v_j)$ then $d(v_i, w)\neq d(v_j, w)$ (i.e, $d(v_i, w)=1, d(v_j, w)\neq 1$) which implies $w\in \Delta_{V}(v_{i}, v_j)$. Similar arguments are true for the case when $w\not\in N(v_i)$ and $w\in N(v_j)$.\\
(ii) Let $v_i, v_j\in T_{j}$ for all $j\in\{2, 3,\cdots T_{n-2}\}$ then by (i) $ N(v_i)\Delta N(v_j)\subset \Delta_{V}(v_{i}, v_j)$ and by (iii) $N(\bar{v_i})\Delta N(\bar{v_j})\subset \Delta_{V}(v_{i}, v_j)$. Now suppose there exist $v_i+v_j, \bar{v_i}+\bar{v_j}\in V$ such that $d(v_i, v_i+v_j )\neq d(v_j, v_i+v_j)$ and $d(v_i, \bar{v_i}+\bar{v_j})\neq d(v_j, \bar{v_i}+\bar{v_j})$ implies that $\{v_i+v_j, \bar{v_i}+\bar{v_j}\}\subset \Delta_{V}(v_{i}, v_j)$.\\
(iii) Let $v_i, v_j\in T_{n-1}$, there exist $w\in V$ such that $w\in N(\bar{v_i})$ and $w\in N(\bar{v_j})$ then $d(\bar{v_i}, w)=1=d(\bar{v_j}, w)$ and $d(v_i, w)=d(v_j, w)$ implies that $w\not\in \Delta_{V}(v_{i}, v_j)$. Now if $w\in N(\bar{v_i})$ and $w\not\in N(\bar{v_j})$ then $d(\bar{v_i}, w)\neq d(\bar{v_j}, w)$ and $d(v_i, w)\neq d(v_j, w)$ implies that $w\in \Delta_{V}(v_{i}, v_j)$. Similar arguments are true for the case when $w\not\in N(\bar{v_i})$ and $w\in N(\bar{v_j})$. At last if $w\not\in N(\bar{v_i})$ and $w\not\in N(\bar{v_j})$ then $d(\bar{v_i}, w)=d(\bar{v_j}, w)$ and $d(v_i, w)=d(v_j, w)$ implies that $w\not\in \Delta_{V}(v_{i}, v_j)$.\\
(iv) Let $ v_i= v_j$, then $d(v_i, k)=d(v_j, k)$ for $k\in V$ implies that $\Delta_{V}(v_{i}, v_j)=\emptyset$.
\end{proof}
\par
Reducts yield a partition such that \( \pi_{\mathbb{A}} = \pi_{V} = v_1 \mid v_2 \mid \cdots \mid v_n \). Now, we explore this concept from a different perspective, focusing on subsets of the vertex set whose removal results in a partition where at least one pair of vertices becomes indiscernible. A set of vertices whose removal alters the partition is referred to as an essential set. Specifically, we start from $\pi_{V}$ and iteratively remove vertices until we obtain a partition different from $\pi_{V}$.
Vertices can be categorized based on their role in the partition structure. Void vertices are those whose removal does not alter the partition, meaning they belong to none of the reducts or resolving sets. In contrast, basis-forced or fixed vertices are those whose removal leads to a change in the partition and these vertices appear in every resolving set or reduct. This terminology of fixed and void vertices is taken from \cite{void}. The smallest set of vertices whose removal yields partition different from $\pi_{V}$ is called essential set. Chiaselotti et al. \cite{Chiaselotti3} introduced the concept of essential sets, also referred to as the extended core, to address cases where the core is empty. In a similar manner, we define essential sets as follows.
\begin{defn}
For $\mathcal{G}$, a set $E \subseteq V$ is an essential set, if $\pi_{V\setminus E}\neq \pi_{V}$ and $\pi_{V\setminus E'}= \pi_{V}$ for all $E'\subset E$.
\end{defn}
We denote by \( ESS(\mathcal{I}) \) as the set containing all essential sets of \( \mathcal{I} \). The subset of \( k \)-essential sets is defined as $ESS_k(\mathcal{I}) = \{E \in ESS(\mathcal{I}): |E| = k\}$. The essential dimension of \( \mathcal{I} \) is given by $E_{dim}(\mathcal{I}) = \min\{k: ESS_k(\mathcal{I}) \neq \emptyset\}.$ By using the collection of all essential sets, we can construct all the minimal resolving sets. The following result presents essential sets for \( \mathcal{I} \) associated with \( \mathcal{G} \), where \( \mathcal{G} \) contains distance-similar vertices.
\begin{Proposition}\label{minimal1}
For $\mathcal{I}$, we have\\
(i) $ESS(\mathcal{I})=\{v_{i}, v_{i}': v_i, v_{i}' \in \mathcal{B}_j\}$ where $j \in \{1, 2, \cdots k\}$.\\
(ii) $E_{dim}(\mathcal{I})=2$.\\
(iii) $|ESS(\mathcal{I})|=\sum_{i=1}^{k}\binom{|\mathcal{B}_i|}{2}$ where $|\mathcal{B}_i|\geq 2$.
\end{Proposition}
\begin{proof}
(i) To prove that $\{v_{i}, v_{i}': v_i, v_{i}' \in \mathcal{B}_j\}$ is an essential set we first show that $\pi_{V}\neq\pi_{V\setminus \{ v_{i}, v_{i}' \}}$. For $v_{i}, v_{i}' \in V$ then from Proposition \ref{prop2} $d(v_i, v_i)=0$ and $d(v_i', v_i)=2$ implies that $\mathcal{F}(v_i, v_i)\neq\mathcal{F}(v_i', v_i)$
gives that $v_i\not\equiv_{V}v_i'$. But $\mathcal{F}(v_i, u)=\mathcal{F}(v_i', u)$ for all $u \in V\setminus \{ v_{i}, v_{i}' \}$ gives $v_i\equiv_{V\setminus \{ v_{i}, v_{i}' \}}v_i'$. So $\pi_{V}\neq\pi_{V\setminus \{ v_{i}, v_{i}' \}}$. By Proposition \ref{prop1} $\pi_{V\setminus \{v_i\}}=\pi_{V}$ hence $\{ v_{i}, v_{i}' \}$ is an essential set. \\
(ii) The proof directly follows from (i).\\
(iii) It is evident from (i) that each subset of cardinality two of each distance-similar class is an essential set. So from each class $\mathcal{B}_i$ with $|\mathcal{B}_i|\geq 2$ there exist $\binom{|\mathcal{B}_i|}{2}$ essential sets. Consequently, the total number of essential sets is given by the sum $\sum_{i=1}^{k}\binom{|\mathcal{B}_i|}{2}$.
\end{proof}
\par
It is interesting to note that, if we remove any element of $DISC(\mathcal{I})$ from $V$ it changes the partition structure, which is similar to removing an essential set from $V$. In \cite{Chiaselotti2} Chiaselotti et al. identified a relationship between the essential sets of an information system and the entries of its corresponding discernibility matrix. They demonstrated that the minimal entries in the discernibility matrix correspond to the essential sets. Analogously, we have the following result for distance-based discernibility matrix.
\begin{Theorem}\label{minimal}
The minimal entries in $DISC(\mathcal{I})$ are essential sets.
\end{Theorem}
\begin{prf}
Suppose $E \in ESS(\mathcal{I})$ then $\pi_{V\setminus E}\neq \pi_{V}$ and $\pi_{V\setminus F}= \pi_{V}$ for all $F\subset E$. Let $v_{i}, v_{j}\in V$ such that
$v_{i}\equiv_{V\setminus E}v_{j}$ but $v_{i}\not \equiv_{V}v_{j}$ thus $E \subseteq \Delta_V(v_{i}, v_{j})$. Since $E \in ESS(\mathcal{I})$ then for every proper subset $F\subset E$, we have $\pi_{V\setminus F}=\pi_V$. That is, $F\notin DISC(\mathcal{I})$. Therefore, $\Delta_V(v_{i}, v_{j}) = E$ and $E$ is a minimal element.
\par
Conversely, assume $E$ is a minimal element $DISC(\mathcal{I})$ then for every proper subset $F\subset E$, we have $F\notin DISC(\mathcal{I})$, which implies $\pi_{V\setminus F}=\pi_V$. Moreover, $E\in DISC(\mathcal{I})$ there exist $v_{i}, v_{j}\in V$ such that $\Delta_V(v_{i}, v_{j}) = E$. Thus, $\pi_V \neq \pi_{V\setminus E}$ which completes the proof.
\end{prf}
\par
If a graph has no twin vertices, then the cardinality of every entry in $DISC(\mathcal{I})$ is at least 3. By the established result \ref{minimal}, all minimal entries in $DISC(\mathcal{I})$ are essential sets. Consequently, the cardinality of any essential set must be greater than or equal to 3, ensuring that no essential set has fewer than three elements.
\begin{Remark}
For $\mathcal{G}$, if $\mathcal{G}$ is twin-free graph then $E_{dim}(\mathcal{I})\geq 3$.
\end{Remark}
\section{Applications in Social Networking}
Networks are fundamental structures used to model and analyze complex systems, particularly in social networking, where they represent relationships among individuals. In this section, we study a practical scenario of networking by using the terminology and concepts introduced in the previous sections.
\subsection{Communication Network}
Consider a network of a data science company, hiring seven individuals for a marketing analytic project, each specializing in different skill set crucial to the company's operations. We model this network with zero-divisor graph, where the nodes are individuals which are represented by a specific skill set. The edge between two nodes represent connections established based on the principle of complementarity in skills and expertise, fostering a collaborative and diverse business environment. In this network, two individuals are connected if they complement each other in skills. There are seven people having three types of expertise: (1) Herry a statistical analyst, (2) Ella possessing skills in both statistical analysis and machine learning, (3) Alice has expertise in machine learning, (4) Bob is a statistical analyst also excels in data visualization, (5) Mark is expert in machine learning and statistical analysis, (6) Jia specializes in statistical analysis, (7) Zoe possesses skills in statistical analysis. She is also expert in data visualization. The network is shown in Figure \ref{fig2}.
\begin{figure}[h!]
  \centering
\includegraphics[width=4cm]{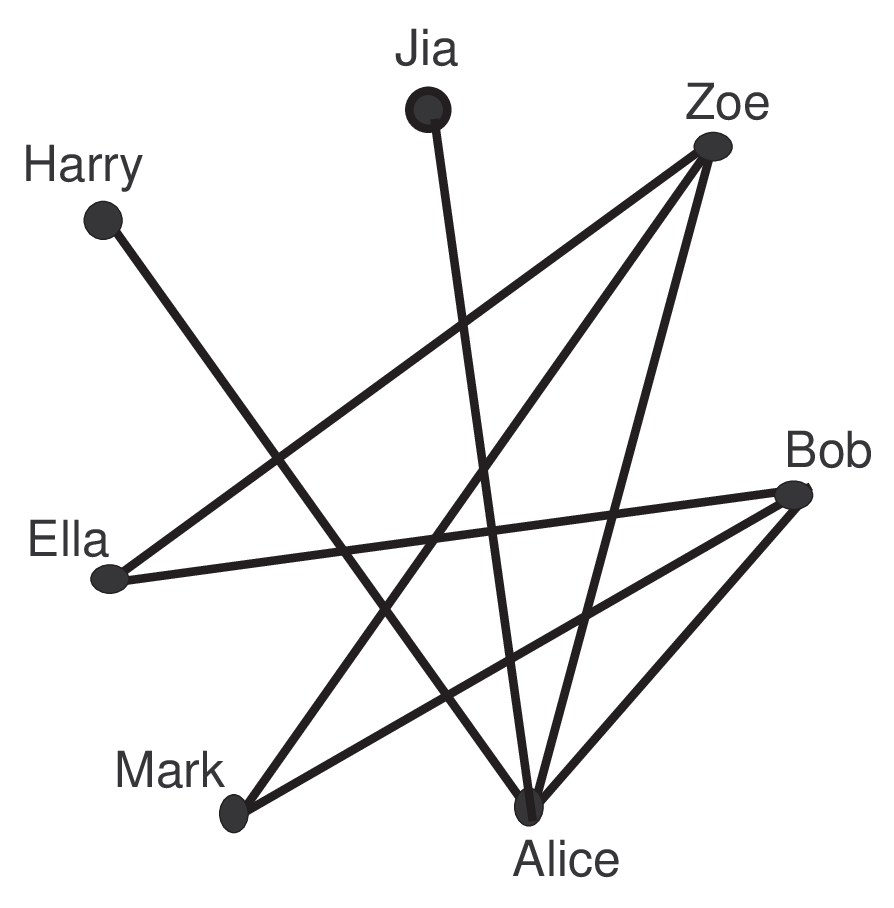}
   \caption{Employees Connected by Diverse skills}%
   \label{fig2}%
\end{figure}
\par
We can interpret the distances as a measure of how complementary the skills of individuals are. The distance matrix is given in Table \ref{Table1}. This network goes in line with the zero-divisor graph of $\mathbb{Z}_{12}$. The vertices with same skill set are considered distance-similar in $\Gamma(\mathbb{Z}_{12})$. From the above network we observe that Harry and Jia, Ella and Mark, Zoe and Bob posses same skill set so they are distance-similar. The distance-similar classes are $\mathcal{B}_1=\{Harry, Jia\}, \mathcal{B}_2=\{Ella, Mark\}, \mathcal{B}_3=\{Zoe, Bob\}, \mathcal{B}_4=\{Alice\}$.

\begin{table}[ht]\label{Table1}
\centering
\scriptsize
\begin{tabular}{|c|c|c|c|c|c|c|c|}
  \hline
   & Herry & Alice & Mark & Bob & Ella &Zoe& Jia\\
\hline
  Herry & 0 & 1 & 3 & 2 & 3& 2&2\\
  Alice &1 & 0 & 2 & 1 & 2& 1& 1\\
  Mark &3 & 2 & 0 & 1 & 2& 1& 3\\
 Bob &2& 1 & 1 & 0 & 1& 2&2\\
  Ella &3 & 2 & 2 & 1 & 0& 1& 3\\
Zoe &2& 1 & 1 & 2 & 1& 0&2\\
Jia & 2 & 1 & 3 & 2 & 3& 2&0\\
\hline
\end{tabular}
\caption{Distance Matrix}
\end{table}
To understand the network's structure and the representation of individuals in terms of their skills and expertise, lets identify resolving sets. A resolving set is a subset of individuals such that the distance between any two individuals in the network is uniquely determined by their distances to the individuals in the resolving set. The task of finding all resolving sets in a network is generally challenging because it involves an exhaustive search over all possible subsets of vertices, and the number of such subsets can be exponential in the size of the network. For larger graphs, this becomes computationally infeasible. Therefore, we use rough set theory to identify all resolving sets.
\par
To compute all resolving sets which are called reducts in the context of RST, we consider distance-based discernibility matrix of Table \ref{Table1}. Since the discernibility matrix is symmetric and has empty diagonal entries, we only need to consider the lower half of the matrix for computations. For simplicity we assume only the first alphabet of their names that is the vertex set $V=\{H, A, M, B, E, Z, J\}$.
\begin{table}[ht]
\scriptsize
\centering
\begin{tabular}{c|ccccccc}
   & Herry & Alice & Mark &Bob & Ella & Zoe &jia\\
\hline
  Herry & $\emptyset$ &  &  &  &  &\\
  Alice& $ V$ & $\emptyset$ &  &  &  &\\
  Mark & $V$& $\{A, M, H, J\}$ & $\emptyset$ &  &  &\\
 Bob & $\{H, M, B, E\}$ & $V$ & $V$ & $\emptyset$ &  &\\
  Ella &  $V$ &  $\{H, A, E, J\}$ & $\{M, E\}$ & $V$ & $\emptyset$ &\\
 Zoe & $\{H, M, E, Z\}$ & $V$ & $V$ & $\{B, Z\}$ & $V$ & $\emptyset$\\
Jia & $\{H,J\}$ & $V$ & $V$ & $\{M, B, E, J\}$ & $V$ &\{M, E, Z, J\}& $\emptyset$\\
\end{tabular}
\end{table}
\par
We determine the reducts using the discernibility function. By applying the disjunction operation to all entries of the discernibility matrix, we obtain $\{M, E\}\wedge\{B, Z\}\wedge\{H, J\}$ and after simplifying, the obtained reducts are $\{M, B, H\}, \{M, Z, H\}, \{M, B, J\}, \{M, Z, J\}, \{E, B, H\}, \{E, B, J\}, \{E, Z, H\}$ and $\{E, Z, J\}$. When finding resolving sets, one often aims to minimize the size of the resolving set while ensuring that all vertices are uniquely determined. All the vertices in distance-similar class is identical so we consider only $m-1$ vertices from each class with cardinality $m$.
\section{Conclusions}
Network analysis helps to understand how different entities are connected and interact in complex systems. Networks are incredibly complex with numerous interconnected components and dynamic behaviors. Granular computing helps in simplifying complex networks by grouping nodes into granules based on their properties or relationships. This paper introduced a metric-based granular computing technique to study network structure by defining an indiscernibility relation on $V$.
The concept of reducts in RST and resolving sets in a network are both concerned with identifying minimal subsets of vertices that give complete information about the network. By establishing the equivalence between reducts and resolving sets, we proposed algorithms to compute minimal resolving sets in networks. The first algorithm identifies resolving sets using the concept of reducts, while the second employs the discernibility function derived from the distance-based discernibility matrix to determine all the minimal resolving sets.
The foundational concepts were applied to simple undirected graphs and extended to zero-divisor graphs associated with $\mathbb{Z}_n$, which include twin vertices, and $\prod_{i=1}^{n}\mathbb{Z}_{2}$, which represent twin-free graphs. In the future, we plan to expand our research by applying the methodology proposed in this paper to study other types of networks.

\section*{Author contributions}
All authors have contributed equally in this paper.
%

\end{document}